\pdfoutput=1
\documentclass[reprint,aps,onecolumn,10pt,superscriptaddress,nofootinbib]{revtex4-2}
\usepackage{graphicx} % Required for inserting images
\usepackage{xcolor} % Required for coloured comments

\usepackage[T1]{fontenc}
\usepackage{enumitem}
\usepackage{tikz}
\usepackage{pgfplots}
\pgfplotsset{compat=1.18}
\usepackage{amsfonts}
\usepackage{amsmath,amssymb,amsthm}
\usepackage{graphicx}
\usepackage{fancyhdr}
\usepackage[breakable]{tcolorbox}
\usepackage{bbm}
\usepackage{url}
\usepackage{mathtools}
\usepackage{physics}
\usepackage{upgreek}
\usepackage{dsfont}
\usepackage{float}
\usepackage{bm}
\usepackage{multirow}
\usepackage{makecell}
\usepackage[strings]{underscore}
\usepackage[ruled,vlined]{algorithm2e}

\usepackage{algpseudocode}

\usepackage{hyperref}
\usepackage{cleveref} % for \cref
\usepackage{aliascnt}

\newtheorem{thm}{Theorem}
\newtheorem*{thm*}{Theorem}
\makeatletter
\newcommand{\setthmtag}[1]{% \settheoremtag{<tag>}
  \let\oldthethm\thethm% Store \thetheorem
  \newcommand{\thethm}{#1}% Redefine it to a fixed value
  \g@addto@macro\endthm{% At \end{theorem}, ...
    \addtocounter{thm}{-1}% ...restore theorem counter value and...
    \global\let\thethm\oldthethm}% ...restore \thetheorem
  }
\makeatother

\newaliascnt{proposition}{thm}
\newtheorem{proposition}[proposition]{Proposition}
\aliascntresetthe{proposition}
\newtheorem*{prop*}{Proposition}
\newtheorem*{lemma*}{Lemma}

\newaliascnt{corollary}{thm}
\newtheorem{cor}[corollary]{Corollary}
\aliascntresetthe{corollary}
\newtheorem*{cor*}{Corollary}

\newaliascnt{conjecture}{thm}

\aliascntresetthe{conjecture}

\newtheorem*{cj*}{Conjecture}
\newtheorem*{Def*}{Definition}

\newaliascnt{definition}{thm}

\aliascntresetthe{definition}

\newaliascnt{question}{thm}

\aliascntresetthe{question}

\newaliascnt{problem}{thm}

\aliascntresetthe{problem}

\newaliascnt{lemma}{thm}

\aliascntresetthe{lemma}

\theoremstyle{definition}
\newtheorem{remark}{Remark}
\newtheorem*{rem*}{Remark}

\tikzset{
    vertex/.style={circle, draw=black, fill=blue, inner sep=0pt, minimum size=6pt}
}

\Crefname{figure}{Fig.}{Figs.}
\Crefname{equation}{Eq.}{Eqs.}
\Crefname{section}{Sec.}{Secs.}
\Crefname{appendix}{Appx.}{Appcs.}
\Crefname{algorithm}{Alg.}{Algs.}
\Crefname{table}{Table}{Tables}

\DeclareMathOperator*{\argmin}{arg\,min}

\DeclareMathOperator{\sgn}{sgn}

\def\bin{\{0, 1\}}

\def\calC{\mathcal{C}}
\def\calF{\mathcal{F}}

\SetKw{KwAnd}{and}
\SetKw{KwOr}{or}
\newsavebox{\boxSCSA}
\savebox\boxSCSA{\raisebox{-8ex}{\includegraphics{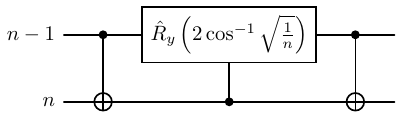}}}
\newsavebox{\boxSCSB}
\savebox\boxSCSB{\raisebox{-8ex}{\includegraphics{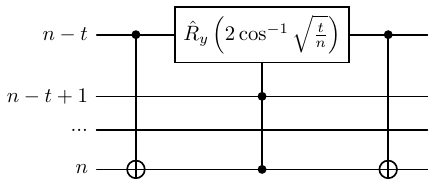}}}

\begin{document}
\pagestyle{plain}
\title{Towards Exponential Quantum Improvements in Solving Cardinality-Constrained Binary Optimization}

\author{\begingroup
    \hypersetup{urlcolor=navyblue}
    \href{https://orcid.org/0009-0001-4941-5448}{Haomu Yuan}
    \endgroup}
\thanks{These authors contributed equally to this work.}
\affiliation{Cavendish Laboratory, Department of Physics, University of Cambridge, Cambridge CB3 0US, UK}
\affiliation{Centre for Quantum Technologies, National University of Singapore, Singapore}
\email[Haomu Yuan ]{hy374@cam.ac.uk}
\author{Hanqing Wu}
\thanks{These authors contributed equally to this work.}
\affiliation{Department of Statistics, Lund University, Lund 220 07, Sweden}
\author{Kuan-Cheng Chen}
\affiliation{Imperial Centre for Quantum Engineering, Science and Technology (QuEST), Imperial College London, London, SW7 2AZ, United Kingdom}
\author{Bin Cheng}
\affiliation{Centre for Quantum Technologies, National University of Singapore, Singapore}
\author{Crispin H. W. Barnes}
\affiliation{Cavendish Laboratory, Department of Physics, University of Cambridge, Cambridge CB3 0US, UK}

\date{\today}
\begin{abstract}
    Cardinality-constrained binary optimization is a fundamental computational primitive with broad applications in machine learning, finance, and scientific computing.
    In this work, we introduce a Grover-based quantum algorithm that exploits the structure of the fixed-cardinality feasible subspace under a natural promise on solution existence.
    For quadratic objectives, our approach achieves $\mathcal{O}\left(\sqrt{\frac{\binom{n}{k}}{{M}}}\right)$ Grover rotations for any fixed cardinality $k$ and degeneracy of the optima $M$, yielding an exponential reduction in the number of Grover iterations compared with unstructured search over $\{0,1\}^n$.
    Building on this result, we develop a hybrid classical--quantum framework based on the alternating direction method of multipliers (ADMM) algorithm.
    The proposed framework is guaranteed to output an $\epsilon$-approximate solution with a consistency tolerance $\epsilon + \delta$ using at most $
\mathcal{O}\left(\sqrt{\binom{n}{k}}\frac{n^{6}k^{3/2} }{ \sqrt{M}\epsilon^2 \delta }\right)$
 queries to a quadratic oracle, together with $\mathcal{O}\left(\frac{n^{6}k^{3/2}}{\epsilon^2\delta}\right)$ classical overhead.
    Overall, our method suggests a practical use of quantum resources and demonstrates an exponential improvements over existing Grover-based approaches in certain parameter regimes, thereby paving the way toward quantum advantage in constrained binary optimization.
\end{abstract}
\maketitle

\section{Introduction}\label{sec:problem_statement}
The binary polynomial programming with fixed-cardinality (BPP-FC) has a wide range of applications in portfolio optimization \cite{markowitzPortfolioSelection1952,Maillard2010,Ricca2024}, biology~\cite{Fratkin2006,Saha2010}, graph and network analysis~\cite{10.1145/792550.792553,10.5555/1083592.1083676, Brandes2005}, and sparse linear regression \cite{bertsimasBestSubsetSelection2016, DedieuLearning2021}. 
For some particular cases of BPP-FC, efficient approximation algorithms exist. One example is the famous dense $k$-subgraph problem, for which Feige \textit{et al}. showed that a greedy algorithm can achieve an approximation ratio $\mathcal{O}(n^\delta)$ (with $\delta < 1/3$) in polynomial time~\cite{Feige2001}. They also note an approximation method to obtain a ratio of $n/k$ using semidefinite programming as subroutines, based on the private communications with Goemans. However, BPP-FC is in general NP-hard to solve. As quantum computing thrives, many quantum algorithms have demonstrated successful speed-ups in quadratic unconstrained binary optimizations (QUBO), such as quantum approximate optimization algorithm (QAOA)~\cite{Farhi2014}, variational quantum eigensolver (VQE)~\cite{Cerezo2021}, quantum annealing~\cite{Kadowaki1998},  Grover adaptive search (GAS)~\cite{Durr1999}, and Hamiltonian Updates~\cite{Yuan2025}. Due to the growing need to demonstrate quantum advantage in practical applications, variants of quantum QUBO solvers for fixed-cardinality problems start to draw more attention in recent years. Those methods can be categorized into two main regimes: noisy intermediate-scale quantum (NISQ) and fault-tolerant quantum computing (FTQC). The NISQ algorithms, including QAOA and VQE, are typically constructed by designing parameterized operators that preserve the structure of the constraint subspace~\cite{Hadfield2019,yuan2024quantifyingadvantagesapplyingquantum}. Meanwhile, in the FTQC regime, Gilliam \textit{et al}. propose a GAS solver for quadratic programming that incorporates a soft constraint into the objective function~\cite{Gilliam2021groveradaptive}. However, rigorous theoretical bounds for speed-ups remain scarce.

In this article, we introduce a new GAS framework for BPP-FC and rigorously quantify its quantum speed-up. Key to this approach is a new design for the diffusion operator in the Grover search subroutine, which confines the search dynamics to the fixed-cardinality subspace for the quadratic programming problem. 
As an application, we extend our method to a specific form of BPP-FC via state-of-the-art alternating direction method of multipliers (ADMM) algorithm to hybridize the classical optimization and our quantum solver. The entire pipeline admits a traceable complexity analysis, confirming the efficient scaling of the required computational resources.
Specifically, in the first part of this article, we consider a general \textit{binary quadratic programming with the fixed-cardinality constraint} (BQP-FC) defined by
\begin{equation}\label{eq:ksubgrap}
    \begin{aligned}
        \argmin_{\mathbf x \in\{0,1\}^n} & \quad f({\mathbf x})= \frac{1}{2}{\mathbf x}^\intercal \boldsymbol\Sigma \mathbf x-{\boldsymbol \mu}^\intercal \mathbf x \\
        \mathrm{s.t.}                             & \quad \mathbf x^\intercal  \mathbf{1} = k,
    \end{aligned}
\end{equation}
where $k \in\{0,\dots,n\}$ is a fixed constant, and $\boldsymbol \mu\in\mathbb{R}^n$ and $\boldsymbol\Sigma\in\mathbb{R}^{n\times n}$ denote the linear and quadratic coefficients of the objective function, respectively.
Note that this optimization admits a direct conversion to a weighted hypergraph {problem} satisfying $K$-locality, where $K$ is the maximum degree of the polynomials. 
Furthermore, due to the fixed-cardinality constraint, \Cref{eq:ksubgrap} is also considered as a generalization of a group of graph optimization problems, also known as the densest $k$-subset problem~\cite{Feige2001, Sotirov2020}, the heaviest unweighted subgraph problem~\cite{Kortsarz1993}, the $k$-cluster problem \cite{CORNEIL198427}, or the $k$-cardinality subgraph problem\cite{Bruglieri2006}.
The problem \Cref{eq:ksubgrap} is, in general, NP-hard for graphs with maximum degree three~\cite{feige1997densest}, and can be solved in polynomial time for graphs with maximum degree two~\cite{CORNEIL198427}.

Additionally, we demonstrate the versatility of our framework by applying it to a challenging quartic optimization problem: the \textit{risk parity model}. This model has exhibited remarkable stability in hedging relative risks in portfolio management~\cite{Maillard2010} and is formulated as:
\begin{equation}
    \label{eq:risk_parity_1}
    \begin{aligned}
        \argmin_{\mathbf x \in \{0,1\}^n} & \quad g({\mathbf x})= \sum_{\substack{1 \leq i, j \leq n \\ i \neq j}} \left(x_i  (\boldsymbol{\Sigma}\mathbf x)_i - x_j (\boldsymbol{\Sigma}\mathbf x)_j\right)^2 + \lambda \left(-\mathbf x^\intercal  \boldsymbol{\mu} + \frac{1}{2} \mathbf x^\intercal  \boldsymbol{\Sigma} \mathbf x\right) \\
        \mathrm{s.t.}                              & \quad \mathbf x^\intercal  \mathbf{1} = k, 
    \end{aligned}
\end{equation}
where $\boldsymbol{\mu}$ is the vector of expected returns, $\boldsymbol{\Sigma}$ is the covariance matrix, and $\lambda$ is a trade-off hyperparameter. The first term in \Cref{eq:risk_parity_1} minimizes the disparity in risk contributions between assets, aiming for an equal-risk allocation. This approach was popularized by Bridgewater Associates' ``All Weather'' strategy, which has delivered robust performance through its inherent stability and diversification~\cite{Qian2005,Asness2012}. However, risk parity problem is quartic optimization, which is in general NP-hard~\cite{Feige2001}, creating significant computational challenges for large-scale portfolios.

The rest of this article is structured as follows. In \Cref{sec:gasinto}, we review Grover adaptive search for solving QUBOs and its Grover search subroutines.
In \Cref{sec:constrained_grover}, we present a novel Grover adaptive search algorithm for solving the BQP-FC.
In \Cref{sec:admm}, we detail the hybrid classical-quantum ADMM framework for solving the risk parity models. In \Cref{sec:resource}, we provide the complete resources analysis of our methods.
Finally, we conclude the paper and discuss future directions in \Cref{sec:conclusion}.

\section{Grover Search and Grover Adaptive Search}\label{sec:gasinto}
In this section, we briefly review the Grover search algorithm~\cite{Grover1996} and the Grover adaptive search (GAS) algorithm~\cite{Durr1999, Bulger2003, Baritompa2005}, which will be used as the quantum subroutine for solving the fixed-cardinality  binary optimization problems in this article.

\paragraph{Grover search.}
Consider the unstructured search problem with a single target.
Let $\vb*{\omega} \in \bin^n$ be the target bitstring and let $f: \bin^n \to \bin$ such that $f(\vb{x}) = 1$ if and only if $\vb{x} = \vb*{\omega}$.
The standard Grover search~\cite{Grover1996} starts with $\ket{s} ={\frac {1}{\sqrt {N}}}\sum _{\vb{x} \in \bin^n} \ket{\vb{x}}$, the uniform superposition over computational basis states.
The probability of observing the target on $\ket{s}$ is $1/N$.
Such a probability is amplified by iteratively alternating between the oracle and diffusion operators.  
Explicitly, the oracle operator is defined as
\begin{equation}
    \hat{O}\ket{\mathbf x} =(-1)^{f(\mathbf x)}\ket{\mathbf x} ,
\end{equation}
and the diffusion operator is defined as
\begin{equation}
    \hat A=2\dyad{s}-\hat I =  \hat H^{\otimes n }(2 \dyad{0^n} -\hat I)  \hat H^{\otimes n } \ ,
\end{equation}
where $\hat H$ is the Hadamard gate, and $I$ is the identity operator. 
The standard Grover algorithm can be written as $r$ repetitions of $\hat A\hat O$:
\begin{equation}
    \ket{\psi_r} = (\hat A\hat O)^{ r } \ket{s},
\end{equation}
where the probability of measuring the $|\boldsymbol \omega\rangle$ on $\ket{\psi_r}$ is $\sin^2((r+\frac{1}{2})\theta)$ and $\theta =2\arcsin {\frac{1}{\sqrt {N}}}$. Assuming $N$ is a large number, $\theta/2$ lies in the interval $[0,\pi/2]$ and is considered negligible when compared with $r \theta$. Thus, when $r\theta \approx \frac{\pi}{2}$, we have $r = \frac{\pi}{2\theta}=\frac{\pi}{4\arcsin(1/\sqrt {N})}\approx \frac{\pi\sqrt {N}}{4}$, and the probability of measuring the $|\boldsymbol \omega\rangle$ is close to 1. 
In the general case with $M$ target solutions, the quantum counting algorithm estimates $M$ with complexity $\mathcal{O}(\sqrt{N/M})$~\cite{Brassard1998}.
\begin{remark}
    We refer to one application of $\hat A \hat O$ as a Grover rotation, since $\hat A\hat O$ acts as a rotation in the two-dimensional subspace spanned by $\ket{\vb*{\omega}}$ and $\ket{\vb*{\omega}^{\perp}} := \frac{1}{\sqrt{N-1}} \sum_{\vb{x} \neq \vb*{\omega}} \ket{\vb{x}}$.
\end{remark}

\paragraph{Grover adaptive search.}
The Grover adaptive search algorithm~\cite{Durr1999, Bulger2003, Baritompa2005} can be used to find the minimum of a QUBO by adaptively updating the oracle used in Grover search.
For instance, given a bounded QUBO objective $f:\{0,1\}^n \to \mathbb{R}$ with $f(\mathbf x)<y'$, we can define an oracle operator $\hat O_y$ as
\begin{equation}
    \begin{aligned}\label{eq:gas_oracle}
        \hat O_y\ket{\mathbf x} ={\sgn(f(\vb{x})-y)} \ket{\vb{x}}
    \end{aligned}
\end{equation}
where $y\leq y'$.
To implement this, we can first encode the two's-complement binary representation of $(f(\mathbf x)-y)\in[-2^{m-1},2^{m-1})$ into $m$ ancilla qubits. As the last bit encodes the sign of the value, the sign function can be implemented by applying a Pauli-$Z$ to that sign bit (controlled on any additional work qubits required by the encoding). Efficient constructions for performing this arithmetic encoding are given in~\cite{Gilliam2021groveradaptive, Gilliam2021} using a structured sequence of rotations and inverse-QFT techniques---see~\Cref{sec:quantum_dictionary_oracle} and~\Cref{tb:oracletable} for details. Additionally, quantum addition circuits and quantum signal processing techniques can be used to construct the sign function oracle~\cite{Cuccaro2004,10821069}.

GAS then iteratively decreases the threshold $y$ based on the outcomes of Grover search subroutines, and eventually converges to the global minimum (which may be degenerate). In this article, we consider the random GAS proposed in~\cite{Durr1999}---see \Cref{alg:random_gas} for details.
For $\xi=1.34$, the expected number of iterations is bounded by $1.32(N/t)^{1/2}$, and the expected number of oracle queries is bounded by $1.32 \sqrt{N} \sum_{r=t+1}^N \frac{1}{r \sqrt{r-1}} \approx 2.46(N/t)^{1/2}$~\cite{Baritompa2005}, where $t$ is the degeneracy of global minimal solutions.

\begin{algorithm}[H]
    \SetAlgoLined
    \textbf{Input: $f: \bin^n \rightarrow \mathbb{R}$, $\xi> 1$, $r_{\max}=1,i =1$ }\\
    \KwResult{$y$}
    Uniformly select $\vb{u}_1\in \bin^n$, set $y_1 = f(\mathbf u_1)$, $\mathbf x= \mathbf u_1$, and $y=y_1$;\\
    \Repeat{a termination condition is met}{
        Uniformly select a rotation count $r_i$ from $\{0,\dots,\lceil r_{\max} -1\rceil\}$;\\
        Perform a Grover search of $r_i$ rotations with oracles $\hat O_{y_i}$, and output $\mathbf x$ and $y$;\\
        \eIf{$y < y_i$}{
            Set $\mathbf u_{i+1} = \mathbf x$, $y_{i+1} = y$, and $r_{\max}=1$;\\
        }{
            Set $\mathbf u_{i+1} = \mathbf u_i$, $y_{i+1} = y_i$, and $r_{\max} = \xi r_{\max}$ \footnotemark;
        }
        $i = i+1$;
    }
    \caption{Grover Adaptive Search (GAS)}\label{alg:random_gas}
\end{algorithm}
\footnotetext{Note that the optimal rotation count for single solution is $\lceil\pi\sqrt{N}/4\rceil$, there will be no need to increase $k$ further. In the framework of our Grover search for hard-constrained quadratic programming, we will provide a more general bound defined by the problem's constraints---see \Cref{sec:constrained_grover} for details.}

\section{Grover adaptive search for binary quadratic programming with fixed-cardinality}
\label{sec:constrained_grover}
In this section, we introduce the Grover search algorithm for hard-constrained programming and use the Grover adaptive search framework to find the minimum of \Cref{eq:ksubgrap}. Compared with the Grover search algorithm for the soft-constrained programming, we show that our method reduces the iterations exponentially. Additionally, we analyze the gate counts and circuit depth of our algorithm.

In~\cite{Gilliam2021groveradaptive}, Gilliam \textit{et al.} proposed a Grover algorithm to solve a soft-constrained~\Cref{eq:ksubgrap} quadratic binary programming by adding a penalty to the objective function, i.e., the Grover algorithm for soft-constrained problems will construct the oracle on a new objective function with the constraint penalty.
Specifically, one will apply Grover algorithm to the following optimization problem
\begin{equation}\label{eq:soft_constrained_grover}
    {\arg\min}_{\mathbf x \in\{0,1\}^n} {\mathbf x}^\intercal \mathbf \Sigma {\mathbf x}-{\boldsymbol \mu}^\intercal {\mathbf x} +\lambda (\sum_i x_i  - k)^2,
\end{equation}
where $\lambda$ is a large number to enforce the satisfaction of the constraint over Grover search. However,  the penalty variable $\lambda$ can only implement a loose bound of the constraint as an overly large value may lead to ill-conditioning of the original objective, and a small value may lead to constraint violation.
Under this assumption, the worst-case time complexity remains bounded by $\mathcal{O}(\sqrt{N/M})$ time. If we consider all combinations of selecting $k$ from $n$ as binomial theorem, there exists a classical constrained brute-force algorithm with a time complexity of $\binom{n}{k}$ for brute-forcing all feasible solutions to the quadratic binary optimization problem with the fixed-cardinality---see Appendix B in \cite{yuan2024quantifyingadvantagesapplyingquantum} for details. Thus, it is unclear whether the soft-constrained Grover algorithm is more efficient than a classical brute-force algorithm. 
In this work, we construct a novel diffusion operator for the hard-constraint problem, extending the framework introduced by \cite{yuan2026quantum}. We note that independent, concurrent work by \cite{bae2026reducingcircuitresourcesgrovers} similarly utilizes Dicke state preparation to enforce fixed-cardinality constraints; however, their approach is developed within the context of standard Grover's search rather than the Grover Adaptive Search (GAS) regime explored here.

In~\Cref{thm:hard_constrained_grover}, we define a Grover search algorithm for hard-constrained problem that guarantees a non-trivial speed-up over all other algorithms.
\begin{thm}[Grover search for hard-constrained programming]\label{thm:hard_constrained_grover}
    Consider a binary quadratic programming subject to a fixed-cardinality constraint with a feasible set $\mathcal{C}$. Let $\mathcal{F}\subseteq \mathcal{C}$ denote the set of solutions with strictly better objective values than those in $\mathcal{C}\backslash \mathcal{F}$ and $|\mathcal{F}|=M$. 
    Assume that there exists an efficient quantum circuit $\hat{C}$ to prepare $\ket{s_c} = \frac{1}{\sqrt{|\calC|}} \sum_{\mathbf x \in \calC} \ket{\mathbf x}$.
    Then, the Grover search for the hard-constrained problem generates a uniform distribution over $\mathcal{F}$, where each solution is measured with probability $\frac{1}{M}$ after approximately $\frac{\pi}{4}\sqrt{\frac{|\mathcal{C}|}{|\mathcal{F}|}}$ iterations.
    In particular, when $\calC$ is defined by the fixed-cardinality $k$, the number of iterations is $\frac{\pi}{4} \sqrt{\binom{n}{k}/M}$.
\end{thm}
\begin{proof}
     The Grover diffusion operator for hard-constrained problems is defined as:
    \begin{equation}
        \hat A_{c} := 2 \dyad{s_c} -\hat{I} = \hat{C} (2 \dyad{0^n} -\hat{I}) \hat{C}^{\dagger}.
    \end{equation}
    Similarly to the original Grover algorithm, the initial state can be written as:
    \begin{equation}
        \ket{s_c} = \cos(a) |\bar{0}\rangle + \sin(a) \ket{\bar{1}},
    \end{equation}
    where
    \begin{equation}\label{eq:constrained_angles}
        \begin{aligned}
            \sin(a)         & = \sqrt{\frac{|\mathcal{F}|}{|\mathcal{C}|}}, & \cos(a)         & = \sqrt{1-\frac{|\mathcal{F}|}{|\mathcal{C}|}}                \\
            \ket{\bar{1}} & = \frac{1}{\sqrt{|\calF|}} \sum_{\mathbf x\in \mathcal{F}} \ket{\mathbf {x}},         & |\bar{0}\rangle & = \frac{1}{\sqrt{|\calC \setminus \calF|}} \sum_{\mathbf x \in \mathcal{C} \setminus   \mathcal{F}} \ket{\mathbf x} .
        \end{aligned}
    \end{equation}
    Thus, when implementing the oracle operators from the GAS~\Cref{eq:gas_oracle}, we will have
    \begin{equation}
        \hat O \ket{s_c} = \cos(a) |\bar{0}\rangle - \sin(a) \ket{\bar{1}}.
    \end{equation}
    Then, by the results of applying $A_c$ on each basis state, i.e.
    \begin{equation}
        \begin{aligned}
            \hat A_c|\bar{0}\rangle & = (2\cos(a)^2-1) |\bar{0}\rangle + 2\sin(a)\cos(a) \ket{\bar{1}} \\
                               & = \cos(2a) |\bar{0}\rangle + \sin(2a) \ket{\bar{1}},             \\
            \hat A_c\ket{\bar{1}} & = (2\sin(a)^2-1) \ket{\bar{1}} + 2\sin(a)\cos(a) |\bar{0}\rangle \\
                               & = -\cos(2a) \ket{\bar{1}} + \sin(2a) |\bar{0}\rangle.
        \end{aligned}
    \end{equation}
    Thus, after one iteration, we have
    \begin{equation}
        \begin{aligned}
            \hat A_c \hat O \ket{s_c} & = [\cos(a)\cos(2a)-\sin(a)\sin(2a)] \ket{\bar{0}} + [\sin(a)\cos(2a)+\cos(a)\sin(2a)] \ket{\bar{1}} \\
                           & = \cos(3a) \ket{\bar{0}} + \sin(3a) \ket{\bar{1}}.
        \end{aligned}
    \end{equation}
    By induction, applying the Grover iterate $\hat A_c \hat O$ for $r$ times, the amplitude of the $\ket{\bar{1}}$ becomes $\sin(a+2ar)$. Thus, we can set $(2r+1)a = \pi/2$ to maximize the amplitude of the $\ket{\bar{1}}$ state to $1$, together with~\Cref{eq:constrained_angles}, it means
    \begin{equation}
        r = \frac{\pi}{4\arcsin(\sqrt{\frac{|\mathcal{F}|}{|\mathcal{C}|}})} \approx \frac{\pi}{4}\sqrt{\frac{|\mathcal{C}|}{|\mathcal{F}|}},
    \end{equation}
    and we derive the conclusion.
\end{proof}
In the following subsections, we will continue to give the design and analysis of diffusion operators and oracles complying with the fixed-cardinality constraint within the framework of~\Cref{thm:hard_constrained_grover}.

\subsection{Diffusion operator}

As shown in~\Cref{thm:hard_constrained_grover}, the core idea of the  Grover search for hard-constrained problems is to construct a diffusion operator that restricts the state space to feasible solutions.
In this subsection, we give an explicit construction of the quantum circuit $\hat{C}$ for generating the following state when $\calC = \left\{ \vb{x} \in \bin^n : |\vb{x}| = k \right\}$ is the cardinality-constrained subset
\begin{equation}\label{eq:constrained_initial_state}
    \ket{s_c} = \hat{C} \ket{0^n} = \frac{1}{\sqrt{|\calC|}} \sum_{\mathbf x \in \calC} \ket{\mathbf x} \ ,
\end{equation}
which will give the construction of the hard-constrained Grover diffusion operator.
We will utilize the quantum circuit for preparing the Dicke state to prepare the state $\ket{s_c}$. 
To distinguish from other constraints, we write $|h_k\rangle$ specifically for the state representing the fixed-cardinality constraints, where $k$ represents the sum value.

A Dicke state preparation is a map, starting from one of the basis states of Hamming weight $|\vb{x}|=k$ to the uniform superposition of all the states with Hamming weight $|\vb{x}|=k$:
\begin{equation}\label{eq:diffusion_Unk}
    \hat U^n_k: \ket{0^{n-k}} \ket{1^k} \rightarrow \ket{h_k}:= \binom{n}{k}^{-\frac{1}{2}} \sum_{\substack{
            \mathbf x \in \{0,1\}^n \\
            |\vb{x}| = k
        }}\ket{\mathbf x} \ .
\end{equation}
A deterministic way to generate such a unitary of this map in the quantum computer can be achieved by an inductive construction~\cite{Baertschi2019}, and is given by:
\begin{equation}\label{eq:unitary_scs}
    \hat U^n_k = \prod_{\ell=2}^k\left(\widehat{S C S}_{\ell, \ell-1} \otimes \mathrm{Id}^{\otimes n-\ell}\right) \cdot \prod_{\ell=k+1}^n\left(\mathrm{Id}^{\otimes \ell-k-1} \otimes \widehat{S C S}_{\ell, k} \otimes \mathrm{Id}^{\otimes n-\ell}\right)
\end{equation}
where $\widehat{SCS}_{i,j}$ is called \textit{Split \& Cyclic Shift} unitary defined as follows.
For all $\ell \in 1, \dots, j$, where $j \leq i-1$
\begin{align}
    \widehat{S C S}_{i, j}\ket{0^{j+1}}                              & =\ket{0^{j+1}}                                                                                                                                      \\
    \widehat{S C S}_{i, j}\ket{0^{j+1-\ell}} \ket{1^{\ell}} & =\sqrt{\frac{\ell}{i}} \ket{0^{j+1-\ell}} \ket{1^{\ell}}+\sqrt{\frac{i-\ell}{i}}\ket{0^{j-\ell}} \ket{1^{\ell}} \ket{0} \\
    \widehat{S C S}_{i, j}\ket{1^{j+1}}                              & =\ket{1^{j+1}}.
\end{align}

Thus, the diffusion operator for the fixed-cardinality constraint is defined as
\begin{equation}
    \hat A_n^k = \hat U^n_k (2 \dyad{0^{n-k} 1^k} -\hat I) \hat U^{n\dagger}_k = 2\ket{h_k}\bra{h_k}-\hat I.
\end{equation}
For the operator $2 \dyad{0^{n-k} 1^k} -\hat{I}$, it can be constructed as follows.
First, note that $2 \dyad{1^n} - \hat{I}$ is the multi-controlled $Z$ gates $\widehat{CZ}_n$ up to a minus sign.
Then, conjugating it by $\hat{X}^{\otimes (n-k)} \hat{I}^{\otimes k}$ gives the quantum circuit for $2 \dyad{0^{n-k} 1^k} -\hat{I}$.
Explicitly, we have 
\begin{equation}\label{eq:diffusion_operator_end}
    \hat A_n^k = \hat{U}^n_k (\hat{X}^{\otimes (n-k)} \hat{I}^{\otimes k}) \widehat{C Z}_n (\hat{X}^{\otimes (n-k)} \hat{I}^{\otimes k}) \hat{U}^{n\dagger}_k \ .
\end{equation}

The $\widehat{SCS}_{l,k}$ will involve the swapping of the qubit excitations on three qubits: the $l$-th qubit, the $(l-k)$-th qubit, and the $(l-k+1)$-th qubit. Thus, the construction of the quantum circuit $\widehat{SCS}_{l,k}$ can be implemented with one two-qubit gate and $k-1$ three-qubit gates---see~\Cref{sec:scs} for details.
From the product structure of $\hat{U}_k^n$ in~\Cref{eq:unitary_scs}, its depth is $\mathcal{O}(n)$, using $\mathcal{O}(kn)$ gates. We also note a parallelization optimization of the circuit blocks inside $\hat U^{n}_k$ reduces the circuit depth to $\mathcal{O}{(k \log \frac{n}{k})}$ if using all-to-all connected circuit~\cite{9951196}.

\subsection{Oracle}
In this article, we use the quantum dictionary(QD) oracle in our BQP-FC quantum solver~\cite{Gilliam2021}---see~\Cref{sec:quantum_dictionary_oracle} for design details and gate usages summarized in~\Cref{tb:oracletable}. Given a quadratic function and as in \Cref{eq:soft_constrained_grover} and $m$-bit precision, the QD oracles requires $mn$ 1-control rotations, and $mn(n-1)/2$ the 2-control rotations to encode the $f(\mathbf x)-y$ to the phase of the prepared state. Additionally, we also need to implement the inverse QFT on the ancillia qubits to convert the phase encoding to the basis encoding, which requires $\mathcal{O}(m\log m)$ Toffoli gates as suggested by~\cite{Nam2020}.

\newcommand{\rowpad}{\rule{0pt}{5.5ex}} 

\begin{table}[!ht]
    \footnotesize
    \centering
    \begin{tabular}{c c c c}
        \hline
        Object & Gate                              & QD-GAS\cite{Gilliam2021groveradaptive} & ADMM-GAS-hard   \\
        \hline
        \multirow{2.3}{*}{\rowpad\makecell{Quadratic                                                          \\~\Cref{eq:ksubgrap}}}
               & \rowpad ${\widehat{C^1R}}(\theta)$            & $mn$                                   & $mn$            \\
               & \rowpad $\widehat{C^2R}(\theta)$            & $m\binom{n}{2}$                        & $m\binom{n}{2}$ \\
        \noalign{\vskip 3pt}\hline
        \multirow{7}{*}{\rowpad\makecell{Quartic                                                              \\~\Cref{eq:risk_parity_1}}}
               & \rowpad $\widehat{C^1R}(\theta)$            & $mn$                                   & $mn$            \\
               & \rowpad $\widehat{C^2R}(\theta)$            & $m\binom{n}{2}$                        & $m\binom{n}{2}$ \\
               & \rowpad $\widehat{C^3R}(\theta)$            & $m\binom{n}{3}$                        & -               \\
               & \rowpad $\widehat{C^4R}(\theta)$            & $m\binom{n}{4}$                        & -               \\
        \noalign{\vskip 3pt}\hline
        \multirow{3}{*}{\rowpad\makecell{$K$-degree                                                           \\non-sparse\\ polynomial\\$K>4$}}
               & \rowpad \makecell{$\widehat{C^1R}(\theta)$} & $mn$                                   & $mn$            \\
               & \rowpad \makecell{$\widehat{C^2R}(\theta)$} & $m\binom{n}{2}$                        & $m\binom{n}{2}$ \\
               & \rowpad \makecell{$\widehat{C^kR}(\theta)$                                                             \\$ k = 3,\dots K$} & $m\binom{n}{k}$& -  \\
        \hline
    \end{tabular}
    \caption[Gate counts per oracle in each Grover search of GAS applied to quadratic, quartic, and $K$-order binary optimization problems.]{Gate counts per oracle in each Grover search of GAS applied to quadratic, quartic, and high-degree binary programmings with fixed-cardinality constraint in addition to $\mathcal{O}(m\log m)$ Toffoli gates for inverse QFT. QD-GAS is short for the quantum dictionary oracle based GAS, and ADMM-GAS-hard is the ADMM-based hybrid GAS approach for hard-constrained programmings. }
    \label{tb:oracletable}
\end{table}

However, for higher-degree polynomial functions, i.e. the quartic objective function in \Cref{eq:risk_parity_1}, the controlled rotation $\widehat{CU}_{\mathbf x}\left(2\pi/{2^m}\right)$ in~\Cref{eq:qd_cux} will require $m\binom{n}{4}$ 4-controlled rotations, $m\binom{n}{3}$ 3-controlled rotations, $m\binom{n}{2}$ 2-controlled rotations and $mn$ 1-controlled rotations. In fact, the increase of polynomial degree will lead to an exponential growth of quantum gates and circuit depth. Thus, it is inefficient to use QD oracle for high-degree polynomial programming directly.
Therefore, a method that can reduce the uses of the quantum gates for high-degree polynomial binary optimization is an necessity for expanding the GAS to broader applications. In the next section, we will introduce an ADMM-based method to decompose the risk parity model's optimization into several quadratic sub-optimizations problems, which can be efficiently solved by our GAS algorithm and classical solvers. Without loss of generality, the proposed ADMM-based method also suits to other BPP-FC instances to reduce quantum resources and circuit depth.

\section{ADMM for the risk-parity problem}\label{sec:admm}
In this section, we introduce the ADMM algorithm for solving the risk parity model in \Cref{eq:risk_parity_1}.
Recent work \cite{gambella2020multiblock} has shown the potential of integrating quantum computing algorithms into the alternating direction method of multipliers (ADMM) \cite{boyd2011distributed} scheme for solving mixed-binary programming problems. Due to the difficulty that current quantum circuits have in solving quartic optimizations, we instead propose the following objective function, which converts the quartic optimization into a bi-quadratic one. The new objective function aligns seamlessly with the scheme of ADMM and has a theoretical convergence guarantee as established in \cite{wang2019global}.

Prior to presenting the ADMM algorithm, we impose the following assumptions regarding the mean vector $\boldsymbol{\mu}$ and covariance matrix $\boldsymbol{\Sigma}$:
\begin{enumerate}
    \item[] \hypertarget{upper_bound}{({\bf Data Boundedness})} There exist constants $c_1, c_2 > 0$ such that $\|\boldsymbol{\mu}\|_{\infty} \le c_1$ and $\Sigma_{ii} \le c_2$ for all $i = 1, \dots, n$.
    \item[] \hypertarget{lower_bound}{({\bf Non-Degeneracy})} There exists a constant $c_3 > 0$ such that $\gamma_{\min}(\boldsymbol{\Sigma}) \ge c_3$.
\end{enumerate}

In the financial setting, both assumptions are mild and reasonable. The assumption of data boundedness reflects the expectation that asset returns and variances are finite and bounded. In particular, this assumption also implies $\gamma_{\max}(\mathbf \Sigma) \le c_2n$ since $\gamma_{\max} \leq \tr(\vb*{\Sigma})$ for a positive semidefinite $\mathbf \Sigma$. Meanwhile, the non-degeneracy assumption is readily satisfied by employing the Ledoit-Wolf shrinkage technique \cite{ledoitHoney2023, ledoitWellconditionedEstimatorLargedimensional2004}, which guarantees a well-conditioned estimate for the sample covariance matrix.

Under the above two assumptions, we can apply ADMM to an equivalent reformulation of the original optimization problem in \Cref{eq:risk_parity_1}:
\begin{equation}\label{eq:quadratic_form}
    \begin{aligned}
        \min_{\mathbf x_1, \mathbf x_2, \mathbf y \in \mathbb R^{n}} & \quad \phi(\mathbf x_1, \mathbf x_2, \mathbf y) :=  g(\mathbf x_1, \mathbf x_2) +  \iota_{\mathcal S}(\mathbf x_1) + f(\mathbf x_2)  + h(\mathbf y) \\
        s.t.                                                         & \quad \mathbf x_1 - \mathbf x_2 - \mathbf y = \mathbf 0,
    \end{aligned}
\end{equation}
where
\begin{equation*}
    \begin{aligned}
        g(\mathbf x_1, \mathbf x_2)
                       & = \sum_{\substack{1 \le i , j \le n \\ i \neq j}} \left(x_{1, i} (\boldsymbol{\Sigma} \mathbf x_2)_i - x_{1, j} (\boldsymbol{\Sigma} \mathbf x_2)_j\right)^2, \\
        f(\mathbf x_2) & =  \lambda\left(-\mathbf x_2^\intercal  \boldsymbol{\mu} + \frac{1}{2}\mathbf x_2^\intercal \boldsymbol \Sigma \mathbf x_2\right),           \\
        h(\mathbf y)   & =  \frac{\zeta}{2} \|\mathbf y\|^2,
    \end{aligned}
\end{equation*}
and $\iota_{\mathcal S}(\cdot)$ is the indicator function of $\mathcal S = \{\mathbf x \in  \{0,1\}^n: \mathbf x^\intercal  \mathbf 1 = k\}$, such that
\begin{equation*}
    \iota_{\mathcal S}(\mathbf x) = \begin{cases}
        0,       & \mathbf x \in \mathcal S \\
        +\infty, & \text{otherwise}.
    \end{cases}
\end{equation*}
Here, $\lambda > 0$ is the trade-off hyperparameter, and $\zeta > 0$ is a sufficiently large constant.

The hybrid ADMM algorithm addresses the optimization problem \Cref{eq:quadratic_form} using the augmented Lagrangian given by
\begin{equation}
    \label{eq:lagrangian}
    \mathcal L_{\beta}(\mathbf x_1, \mathbf x_2, \mathbf y, \mathbf w) = \phi(\mathbf x_1, \mathbf x_2, \mathbf y) + \mathbf w^\intercal\left(\mathbf x_1 - \mathbf x_2 - \mathbf y\right) + \frac{\beta}{2}\|  \mathbf x_1 - \mathbf x_2 - \mathbf y\|^2.
\end{equation}

We define the updates at the $(t+1)$-th iteration as $(\mathbf x_1^+, \mathbf x_2^+, \mathbf y^+, \mathbf w^+) = (\mathbf x_1^{t+1}, \mathbf x_2^{t+1}, \mathbf y^{t+1}, \mathbf w^{t+1})$ for notational convenience. The proposed method is summarized in Algorithm \ref{alg:admm}.

\begin{algorithm}[!htbp]\label{alg:admm}
    \SetAlgoLined
    % \textbf{Input:} $t=0, \mathbf{x_1}^0 \in {\mathcal{S}}$, $\mathbf x_2^0 = \mathbf x_1^0$, $\mathbf y^0 = \mathbf w^0 = \mathbf 0$, \text{the vector of return} $\boldsymbol{\mu}$, \text{the covariance matrix} $\boldsymbol{\Sigma}$, the trade-off hyperparameter $\lambda$, the tolerance hyperparameters $\epsilon >0$, the penalty coefficient $\zeta > 0$, the penalty coefficient $\beta > \sqrt{2}\zeta$, the maximum number of iterations $T_{\max}$, and initialize the error term $\Delta = +\infty$.
    \textbf{Input:} $t=0, \mathbf{x}_1^0 \in {\mathcal{S}}$, $\mathbf x_2^0 = \mathbf x_1^0$, $\mathbf y^0 = \mathbf w^0 = \mathbf 0$, \text{the vector of return} $\boldsymbol{\mu}$, \text{the covariance matrix} $\boldsymbol{\Sigma}$, the trade-off hyperparameter $\lambda$, the tolerance hyperparameter $\epsilon >0$, the regularization parameter $\zeta > 0$, the augmented Lagrangian penalty parameter $\beta > \sqrt{2}\zeta$ (as required by \Cref{coro:sufficient_descent}), the maximum number of iterations $T_{\max}$, and initialize the error term $\Delta = +\infty$.

    \KwResult{$\mathbf{x}_1^t$}
    \While{$t < T_{\max}$ \KwAnd  $\Delta \ge \dfrac{\epsilon}{\beta + 1}$}{
        \   \\
        \begin{enumerate}
            \item Use the efficient quantum optimization algorithm to solve the binary optimization subproblem:
                  \begin{equation}\label{eq:quadratic_subproblem_quantum}
                      \begin{aligned}
                          \mathbf x_1^{+} & = \argmin_{\mathbf x_1\in \mathcal S}  \ \sum_{\substack{1 \le i , j \le n \\ i \neq j}} \left(x_{1, i} (\boldsymbol{\Sigma} \mathbf x_2^t)_i  - x_{1, j} (\boldsymbol{\Sigma} \mathbf x_2^t)_j\right)^2 \\
                                          & \quad \quad + \mathbf w^t{}^\intercal \mathbf x_1 + \frac{\beta}{2}\|\mathbf x_1 - \mathbf x_2^t - \mathbf y^t\|^2
                      \end{aligned}
                  \end{equation}
            \item Solve the convex optimization subproblem:
                  \begin{equation}\label{eq:admm2}
                      \begin{aligned}
                          \mathbf x_2^{+} & = \argmin_{\mathbf x_2\in\mathbb R^n}   \ \sum_{\substack{1 \le i , j \le n \\ i \neq j}} \left(x_{1, i}^+ (\boldsymbol{\Sigma} \mathbf x_2)_i - x_{1, j}^+ (\boldsymbol{\Sigma} \mathbf x_2)_j\right)^2                                                             \\
                                          & \quad \quad + \lambda\left(-\mathbf x_2^\intercal  \boldsymbol{\mu} + \frac{1}{2}\mathbf x_2^\intercal \boldsymbol \Sigma \mathbf x_2\right) - \mathbf w^t{}^\intercal \mathbf x_2 + \frac{\beta}{2}\|\mathbf x_1^+ - \mathbf x_2 - \mathbf y^t\|^2
                      \end{aligned}
                  \end{equation}
            \item Solve the convex optimization subproblem:
                  \begin{equation}\label{eq:admm3}
                      \begin{aligned}
                          \mathbf y^{+} & = \argmin_{\mathbf y\in\mathbb R^n} \ \frac{\zeta}{2}\|\mathbf y\|^2 - \mathbf w^t{}^\intercal \mathbf y + \frac{\beta}{2}\|\mathbf x_1^+ - \mathbf x_2^+ - \mathbf y\|^2 \\
                      \end{aligned}
                  \end{equation}
            \item Update the dual variable: $\mathbf w^+ = \mathbf w^t + \beta(\mathbf x_1^+ - \mathbf x_2^+ - \mathbf y^+) $.
            \item Update the error term: $\Delta = \|\mathbf y^+ - \mathbf y^{t}\|$
            \item Set $t = t+1$
        \end{enumerate}
    }
    \Return $\mathbf{x}_1^{t}$
    \caption{Hybrid ADMM algorithm for solving risk parity problem~\Cref{eq:quadratic_form}}
\end{algorithm}

Both subproblems \Cref{eq:admm2} and \Cref{eq:admm3} can be solved explicitly. Specifically, for \Cref{eq:admm2}, note that
\begin{equation}
    \label{eq:matrix_Q}
    \begin{aligned}
        \sum_{\substack{1 \le i , j \le n \\ i \neq j}} \left(x_{1, i}^+ (\boldsymbol{\Sigma} \mathbf x_2)_i - x_{1, j}^+ (\boldsymbol{\Sigma} \mathbf x_2)_j\right)^2 & = \sum_{i=1}^n 2(n-1)x_{1, i}^+{}^2 (\boldsymbol{\Sigma} \mathbf x_2)_i^2 - 2\sum_{\substack{1 \le i , j \le n \\ i \neq j}} x_{1, i}^+x_{1, j}^+ (\boldsymbol{\Sigma} \mathbf x_2)_i (\boldsymbol{\Sigma} \mathbf x_2)_j \\
                                                                                                                                                      & = \mathbf x_2^\intercal \mathbf Q \mathbf x_2,
    \end{aligned}
\end{equation}
where $\mathbf Q = \boldsymbol\Sigma \mathbf H \boldsymbol\Sigma$ is a positive semidefinite matrix, and $\mathbf H$ is a $n\times n$ symmetric matrix such that
\begin{equation*}
    H_{ij} = \begin{cases}
        2(n-1)x_{1, i}^+{}^2   & i = j    \\
        -2x_{1, i}^+x_{1, j}^+ & i \ne j.
    \end{cases}
\end{equation*}
As a result,
\begin{equation*}
    \begin{aligned}
        \mathbf x_2^+ & =  \argmin_{\mathbf x_2\in\mathbb R^n} \mathbf x_2^\intercal\mathbf Q\mathbf x_2 + \lambda\left(-\mathbf x_2^\intercal  \boldsymbol{\mu} + \frac{1}{2}\mathbf x_2^\intercal \boldsymbol \Sigma \mathbf x_2\right) - \mathbf w^t{}^\intercal \mathbf x_2 + \frac{\beta}{2}\|\mathbf x_1^+ - \mathbf x_2 - \mathbf y^t\|^2 \\
                      & = \left(2\mathbf Q + \lambda\mathbf \Sigma + \beta\mathbf I\right)^{-1}(\lambda\mathbf \mu + \mathbf w^t +  \beta(\mathbf x_1^+ - \mathbf y^t)).
    \end{aligned}
\end{equation*}
For \Cref{eq:admm3},
\begin{equation}
    \label{eq:y_update}
    \mathbf y^{+} = \frac{1}{\zeta + \beta}\left(\mathbf w^{t} + \beta(\mathbf x_1^+ - \mathbf x_2^+) \right).
\end{equation}

The work in \cite{wang2019global} analyzes a quite general setting for non-convex non-smooth optimization via ADMM under a series of regularity assumptions. In our specific context of risk parity, we directly verify the convergence of the proposed Algorithm \ref{alg:admm} by establishing the following four key properties (indeed, our optimization problem satisfies those regularity assumptions as well):
\begin{enumerate}
    \item[P1] \hypertarget{P1}{({\bf Boundedness of the Sequence})} The sequence $\left\{\mathbf{x}_1^t, \mathbf x_2^t, \mathbf y^t, \mathbf w^t\right\}$ is bounded, and $ \mathcal L_{\beta}(\mathbf x_1^t, \mathbf x_2^t, \mathbf y^t, \mathbf w^t)$ is lower bounded.
    \item[P2] \hypertarget{P2}{({\bf Sufficient Descent})} There is a constant $C_1(\beta)>0$ such that for all $t$, we have
          \begin{equation*}
              \mathcal L_{\beta}(\mathbf x_1^t, \mathbf x_2^t, \mathbf y^t, \mathbf w^t) - \mathcal L_{\beta}(\mathbf x_1^+, \mathbf x_2^+, \mathbf y^+, \mathbf w^+) \ge C_1(\beta) \left(\|\mathbf x_2^+ - \mathbf x_2^t\|^2 + \|\mathbf y^+ - \mathbf y^t\|^2\right).
          \end{equation*}
    \item[P3] \hypertarget{P3}{({\bf Subgradient bound})} There exists $C_2(\beta)>0$ and $\mathbf d^{+} \in \partial \mathcal L_{\beta}(\mathbf x_1^{+}, \mathbf x_2^{+}, \mathbf y^{+}, \mathbf w^{+})$ such that
          \begin{equation*}
              \|\mathbf d^{+}\| \le C_2(\beta)\left(\|\mathbf x_2^+ - \mathbf x_2^t\| + \|\mathbf y^+ - \mathbf y^t\|\right).
          \end{equation*}
    \item[P4] \hypertarget{P4}{({\bf Limiting continuity})} If $(\mathbf{x}_1^*, \mathbf x_2^*, \mathbf y^*, \mathbf w^*)$ is the limit point of a sub-sequence $(\mathbf{x}_1^{t_s}, \mathbf x_2^{t_s}, \mathbf y^{t_s}, \mathbf w^{t_s})$ for $s \in \mathbb{N}$, then $\mathcal{L}_\beta(\mathbf{x}_1^*, \mathbf x_2^*, \mathbf y^*, \mathbf w^*)=\lim _{s \rightarrow \infty} \mathcal{L}_\beta(\mathbf{x}_1^{t_s}, \mathbf x_2^{t_s}, \mathbf y^{t_s}, \mathbf w^{t_s})$.
\end{enumerate}

We verify the four properties in the order of P2 $\rightarrow$ P1 $\rightarrow$ P3 $\rightarrow$ P4. To establish the sufficient descent property \hyperlink{P2}{P2}, we first present the following two propositions:
\begin{proposition}[cf. Lemma 4 in \cite{wang2019global}]
    \label{prop:suff_descent_1}
    The iterates in Algorithm \ref{alg:admm} satisfy
    \begin{equation}
        \label{eq:suff_descent_1}
        \mathcal L_{\beta}(\mathbf x_1^t, \mathbf x_2^t, \mathbf y^t, \mathbf w^t) -  \mathcal L_{\beta}(\mathbf x_1^+, \mathbf x_2^+, \mathbf y^t, \mathbf w^t) \ge \frac{\beta}{2}\|\mathbf x_2^+ - \mathbf x_2^t\|^2.
    \end{equation}
\end{proposition}

\begin{proof}
    The update of Algorithm \ref{alg:admm} implies that
    \begin{equation}
        \begin{aligned}
            \label{eq:succ_diff_1}
            \mathcal L_{\beta}(\mathbf x_1^t, \mathbf x_2^t, \mathbf y^t, \mathbf w^t) -  \mathcal L_{\beta}(\mathbf x_1^+, \mathbf x_2^+, \mathbf y^t, \mathbf w^t)
             & \ge  \mathcal L_{\beta}(\mathbf x_1^+, \mathbf x_2^t, \mathbf y^t, \mathbf w^t) -  \mathcal L_{\beta}(\mathbf x_1^+, \mathbf x_2^+, \mathbf y^t, \mathbf w^t)         \\
             & = f(\mathbf x_2^t) - f(\mathbf x_2^+) + g(\mathbf x_1^+, \mathbf x_2^t) - g(\mathbf x_1^+, \mathbf x_2^+) - \mathbf w^t{}^\intercal (\mathbf x_2^t - \mathbf x_2^+) + \\
             & \hphantom{ = } + \frac{\beta}{2}\|\mathbf x_1^+ - \mathbf x_2^t - \mathbf y^t\|^2 - \frac{\beta}{2}\|\mathbf x_1^+ - \mathbf x_2^+ - \mathbf y^t\|^2
        \end{aligned}
    \end{equation}
    While the optimality conditon for subproblem of $\mathbf x_2$, $\nabla_{\mathbf x_2} \mathcal L_\beta(\mathbf x_1^+, \mathbf x_2, \mathbf y^t, \mathbf w^t) = \mathbf 0$ implies
    \begin{equation}
        \label{eq:optimal_x2}
        \nabla f(\mathbf x_2^+) + \nabla_2 g(\mathbf x_1^+, \mathbf x_2^+) = \mathbf w^t + \beta(\mathbf x_1^+ - \mathbf x_2^+ - \mathbf y^t)
    \end{equation}
    Combining \Cref{eq:succ_diff_1} and \Cref{eq:optimal_x2}, we have that
    \begin{equation}
        \label{eq:succ_diff_2}
        \begin{aligned}
            \mathcal L_{\beta}(\mathbf x_1^+, \mathbf x_2^t, \mathbf y^t, \mathbf w^t) -  \mathcal L_{\beta}(\mathbf x_1^+, \mathbf x_2^+, \mathbf y^t, \mathbf w^t) & =
            f(\mathbf x_2^t) - f(\mathbf x_2^+) -  [\nabla f(\mathbf x_2^+)]^\intercal (\mathbf x_2^t - \mathbf x_2^+) + g(\mathbf x_1^+, \mathbf x_2^t) - g(\mathbf x_1^+, \mathbf x_2^+)                                                                                                                                           \\
                                                                                                                                                                     & \hphantom{ = }  - [\nabla_2 ~g(\mathbf x_1^+, \mathbf x_2^+)]^\intercal (\mathbf x_2^t - \mathbf x_2^+) + \frac{\beta}{2}\|\mathbf x_2^t - \mathbf x_2^+ \|^2 \\
        \end{aligned}
    \end{equation}
    Specifically, since in our model, $f(\mathbf x_2)$ is convex and differentiable, we have that
    \begin{equation}
        \label{eq:convexity_f}
        f(\mathbf x_2^t) - f(\mathbf x_2^+) -  [\nabla f(\mathbf x_2^+)]^\intercal(\mathbf x_2^t - \mathbf x_2^+) \ge 0.
    \end{equation}
    Additionally, let $\hat g(\mathbf x_2) = g(\mathbf x_1^+, \mathbf x_2)$, then $\nabla_2 ~g(\mathbf x_1^+, \mathbf x_2) = \nabla \hat g(\mathbf x_2)$. According to \Cref{eq:matrix_Q}, $\hat g(\mathbf x_2) = \mathbf x_2^\intercal \mathbf Q \mathbf x_2 $, where $\mathbf Q$ is a PSD matrix. This implies that $\hat g(\mathbf x_2)$ is also a convex differentiable function, and
    \begin{equation}
        \label{eq:convexity_g}
        \hat g(\mathbf x_2^t) - \hat g(\mathbf x_2^+) -  [\nabla g(\mathbf x_2^+)]^\intercal(\mathbf x_2^t - \mathbf x_2^+) \ge 0
    \end{equation}
    Combining \Cref{eq:succ_diff_2}, \Cref{eq:convexity_f}, and \Cref{eq:convexity_g}, it is immediate that
    \begin{equation*}
        \mathcal L_{\beta}(\mathbf x_1^+, \mathbf x_2^t, \mathbf y^t, \mathbf w^t) -  \mathcal L_{\beta}(\mathbf x_1^+, \mathbf x_2^+, \mathbf y^t, \mathbf w^t) \ge \frac{\beta}{2}\|\mathbf x_2^+ - \mathbf x_2^t \|^2.
    \end{equation*}
    The proposition is hereby proved.
\end{proof}

\begin{proposition}[cf. Lemma 5 in \cite{wang2019global}]
    \label{prop:suff_descent_2}
    The iterates in Algorithm \ref{alg:admm} satisfy
    \begin{equation}
        \label{eq:suff_descent_2}
        \mathcal L_{\beta}(\mathbf x_1^+, \mathbf x_2^+, \mathbf y^t, \mathbf w^t) -  \mathcal L_{\beta}(\mathbf x_1^+, \mathbf x_2^+, \mathbf y^+, \mathbf w^+) \ge \left(\frac{\beta}{2} - \frac{\zeta^2}{\beta}\right)\|\mathbf y^+ - \mathbf y^t\|^2
    \end{equation}
\end{proposition}

\begin{proof}
    In our context,
    \begin{equation}
        \label{eq:succ_diff_3}
        \begin{aligned}
            \mathcal L_{\beta}(\mathbf x_1^+, \mathbf x_2^+, \mathbf y^t, \mathbf w^t) -  \mathcal L_{\beta}(\mathbf x_1^+, \mathbf x_2^+, \mathbf y^+, \mathbf w^+)
             & = h(\mathbf y^t) - h(\mathbf y^+) + (\mathbf w^t)^\intercal(\mathbf x_1^+ - \mathbf x_2^+ - \mathbf y^t) - (\mathbf w^+)^\intercal(\mathbf x_1^+ - \mathbf x_2^+ - \mathbf y^+) \\
             & \hphantom{ = }+ \frac{\beta}{2}\|\mathbf x_1^+ - \mathbf x_2^+ - \mathbf y^t\|^2 -  \frac{\beta}{2}\|\mathbf x_1^+ - \mathbf x_2^+ - \mathbf y^+\|^2
        \end{aligned}
    \end{equation}
    Since $\mathbf w^+ - \mathbf w^t = \beta(\mathbf x_1^+ - \mathbf x_2^+ - \mathbf y^+)$ and $\zeta \mathbf y^+  = \mathbf w^+$ (from \Cref{eq:y_update}), \Cref{eq:succ_diff_3} can be reformulated as
    \begin{equation}
        \label{eq:succ_diff_4}
        \mathcal L_{\beta}(\mathbf x_1^+, \mathbf x_2^+, \mathbf y^t, \mathbf w^t) -  \mathcal L_{\beta}(\mathbf x_1^+, \mathbf x_2^+, \mathbf y^+, \mathbf w^+) = h(\mathbf y^t) - h(\mathbf y^+) - (\mathbf w^+)^\intercal (\mathbf y^t - \mathbf y^+) + \left(\frac{\beta}{2} - \frac{\zeta^2}{\beta}\right)\|\mathbf y^t - \mathbf y^+\|^2.
    \end{equation}
    The optimality condition for subproblem of $\mathbf y$, $\nabla_{\mathbf y} \mathcal L_\beta(\mathbf x_1^+, \mathbf x_2, \mathbf y^t, \mathbf w^t) = \mathbf 0$ implies
    \begin{equation}
        \label{eq:nabla_y_w}
        \nabla h(\mathbf y^+) - \mathbf w^t - \beta(\mathbf x_1^+ - \mathbf x_2^+ - \mathbf y^+) = \nabla h(\mathbf y^+) - \mathbf w^+ = 0.
    \end{equation}
    Plugging \Cref{eq:nabla_y_w} into \Cref{eq:succ_diff_4}, we have
    \begin{equation*}
        \mathcal L_{\beta}(\mathbf x_1^+, \mathbf x_2^+, \mathbf y^t, \mathbf w^t) -  \mathcal L_{\beta}(\mathbf x_1^+, \mathbf x_2^+, \mathbf y^+, \mathbf w^+) = h(\mathbf y^t) - h(\mathbf y^+) - [\nabla h(\mathbf y^+)]^\intercal (\mathbf y^t - \mathbf y^+) + \left(\frac{\beta}{2} - \frac{\zeta^2}{\beta}\right)\|\mathbf y^t - \mathbf y^+\|^2.
    \end{equation*}
    However, since $h(\mathbf y) = \frac{\zeta}{2}\|\mathbf y\|^2$ is a convex function, we must have $ h(\mathbf y^t) - h(\mathbf y^+) - [\nabla h(\mathbf y^+)]^\intercal (\mathbf y^t - \mathbf y^+) \ge 0$. As a result, \Cref{eq:suff_descent_2} holds.
\end{proof}

Combining Proposition~\ref{prop:suff_descent_1} and ~\ref{prop:suff_descent_2}, we immediately have the following corollary for the sufficient descent property \hyperlink{P2}{P2}:
\begin{cor}
    \label{coro:sufficient_descent}
    Let $\beta > \sqrt{2}\zeta$, the sufficient descent property \hyperlink{P2}{P2} holds with $C_1(\beta) = \dfrac{\beta}{2} - \dfrac{\zeta^2}{\beta}$.
\end{cor}

Now we turn to the property of the boundedness of the sequence \hyperlink{P1}{P1}. To validate this property, we first prove the coercivity (or level-boundedness in the terminology of \cite{rockafellarVariationalAnalysis1998}) of objective function, stated in the following proposition:

\begin{proposition}[Coercivity]
    \label{prop:coercivity}
    Define the feasible set $\mathcal{F}:=\left\{(\mathbf x_1, \mathbf x_2, \mathbf y) \in \mathbb{R}^{3n}: \mathbf x_1 - \mathbf x_2 - \mathbf y=0\right\}$. The objective function $\phi(\mathbf x_1, \mathbf x_2, \mathbf y)$ is coercive over this set, that is, $\phi(\mathbf x_1, \mathbf x_2, \mathbf y) \rightarrow \infty$ if $(\mathbf x_1, \mathbf x_2, \mathbf y) \in \mathcal{F}$ and $\|(\mathbf x_1, \mathbf x_2, \mathbf y)\| \rightarrow \infty$.
\end{proposition}

\begin{proof}
    If $\|\mathbf x_1 \| \to +\infty$, then by definition of $\iota_{\mathcal S}(\mathbf x_1)$ we immediately have $\phi(\mathbf x_1, \mathbf x_2, \mathbf y) \to +\infty$. Therefore, we only need to consider the case when  $\|(\mathbf x_1, \mathbf x_2, \mathbf y)\| \to +\infty$ with $\|\mathbf x_1\|$ bounded. Given the constraint, we must have that $\|\mathbf x_2\| \to +\infty$ and $\|\mathbf y\| \to +\infty$.

    Both $g(\mathbf x_1, \mathbf x_2)$ and $\iota_{\mathcal S}(\mathbf x_1)$ are non-negative. With the covariance matrix $\boldsymbol{\Sigma}$ being positive definite, we have
    \begin{equation*}
        f(\mathbf x_2) \ge \lambda\left(\dfrac{\gamma_{\min}(\boldsymbol{\Sigma})}{2} \|\mathbf x_2\|^2 - \|\boldsymbol \mu\| \|\mathbf x_2\|\right) \to +\infty \quad \text{as } \|\mathbf x_2\| \to +\infty,
    \end{equation*}
    where $\gamma_{\min}(\cdot)$ refers to the smallest eigenvalue of the corresponding matrix. It is straightforward that $ \frac{\zeta}{2} \|\mathbf y\|^2 \to +\infty$ as $\|\mathbf y\| \to +\infty$. In summary, the coercivity condition holds.
\end{proof}

We now prove property \hyperlink{P1}{P1} holds for our optimization problem:
\begin{proposition}[Boundedness of the sequence, cf. Lemma 6 in \cite{wang2019global}]
    \label{prop:boundedness}
    When $\beta > \zeta$, the sequence $\left\{\mathbf{x}_1^t, \mathbf x_2^t, \mathbf y^t, \mathbf w^t\right\}$ generated by Algorithm \ref{alg:admm} satisfies Property \hyperlink{P1}{P1}.
\end{proposition}

\begin{proof}
    We first show that $\inf_{(\mathbf x_1, \mathbf x_2, \mathbf y)\in \mathcal F} \phi(\mathbf x_1, \mathbf x_2, \mathbf y) > -\infty$. According to Theorem 1.9 in \cite{rockafellarVariationalAnalysis1998}, if an extended-real-valued function $f:\mathbb R^m \to \overline{\mathbb R}$ is lower-semicontinuous, level-bounded, and proper, $\inf f$ is finite. In our case, we consider the extended objective function $\Phi(\mathbf x_1, \mathbf x_2, \mathbf y) = \phi(\mathbf x_1, \mathbf x_2, \mathbf y) + \iota_{\mathcal F}(\mathbf x_1, \mathbf x_2, \mathbf y)$. The property of level-boundedness follows from Proposition \ref{prop:coercivity} and it is immediate that $\Phi(\mathbf x_1, \mathbf x_2, \mathbf y)$ is a proper function (i.e., it is not identically $+\infty$ and never takes the value $-\infty$). It remains for us to prove that $\Phi(\mathbf x_1, \mathbf x_2, \mathbf y)$ is lower semicontinuous, which reduces to showing the lower semicontinuity of $\iota_{\mathcal S}(\mathbf x_1)$ and $\iota_{\mathcal F}(\mathbf x_1, \mathbf x_2, \mathbf y)$, since $g(\mathbf x_1, \mathbf x_2), f(\mathbf x_2)$, and $h(\mathbf y)$ are all continuous functions. Recall that a function is lower semicontinuous if the preimage of any open interval $(u, \infty]$ is an open set. For any $u \ge 0$, we know that $\iota_{\mathcal S}^{-1}((u, \infty]) = \mathbb R^n \setminus \mathcal S$ and $\iota_{\mathcal F}^{-1}((u, \infty]) = \mathbb R^{3n} \setminus \mathcal F$ which are both open sets as $\mathcal S$ and $\mathcal F$ are both closed. Whereas for any $u < 0$, $\iota_{\mathcal S}^{-1}((u, \infty]) =\mathbb R^n$ and $\iota_{\mathcal F}^{-1}((u, \infty]) = \mathbb R^{3n}$ which are also both open. Therefore, we have that $\iota_{\mathcal S}(\mathbf x_1)$ and $\iota_{\mathcal F}(\mathbf x_1, \mathbf x_2, \mathbf y)$ are both lower semicontinuous. As a result, we know that $\Phi(\mathbf x_1, \mathbf x_2, \mathbf y)$ is lower semicontinuous and so $\inf_{(\mathbf x_1, \mathbf x_2, \mathbf y)\in \mathcal F} \phi(\mathbf x_1, \mathbf x_2, \mathbf y) > -\infty$.

    We now show the lower boundedness of $\mathcal L_{\beta}(\mathbf x_1^t, \mathbf x_2^t, \mathbf y^t, \mathbf w^t)$.
    \begin{equation}
        \label{eq:bound_Lbeta}
        \begin{aligned}
            \mathcal L_{\beta}(\mathbf x_1^t, \mathbf x_2^t, \mathbf y^t, \mathbf w^t)
             & = \phi(\mathbf x_1^t,  \mathbf x_2^t, \mathbf y^t) + \mathbf w^t{}^\intercal(\mathbf x_1^t - \mathbf x_2^t - \mathbf y^t) + \frac{\beta}{2}\|\mathbf x_1^t - \mathbf x_2^t - \mathbf y^t\|^2                                                                                    \\
             & \ge \phi(\mathbf x_1^t, \mathbf x_2^t, \mathbf x_1^t - \mathbf x_2^t) + h(\mathbf y^t) - h(\mathbf x_1^t - \mathbf x_2^t) + [\nabla h(\mathbf y^t)]^\intercal(\mathbf x_1^t - \mathbf x_2^t - \mathbf y^t)  + \frac{\beta}{2}\|\mathbf x_1^t - \mathbf x_2^t - \mathbf y^t\|^2.
        \end{aligned}
    \end{equation}
    It is straightforward that
    \begin{equation}
        \label{eq:inf_phi}
        \phi(\mathbf x_1^t, \mathbf x_2^t, \mathbf x_1^t - \mathbf x_2^t) \ge \inf_{(\mathbf x_1, \mathbf x_2, \mathbf y) \in\mathcal F} \phi(\mathbf x_1, \mathbf x_2, \mathbf y) > -\infty.
    \end{equation}
    On the other hand, since $\nabla h(y) = \zeta \mathbf y$ has a Lipschitz constant $\zeta$, we have (see e.g., Theorem 2.1.5 in \cite{nesterovLecturesConvexOptimization2018})
    \begin{equation}
        \label{eq:lipschitz_lower}
        h(\mathbf y^t) - h(\mathbf x_1^t - \mathbf x_2^t) + [\nabla h(\mathbf y^t)]^\intercal(\mathbf x_1^t - \mathbf x_2^t - \mathbf y^t) \ge -\frac{\zeta}{2}\|\mathbf x_1^t - \mathbf x_2^t - \mathbf y^t\|^2.
    \end{equation}
    Plugging \Cref{eq:inf_phi} and \Cref{eq:lipschitz_lower} into \Cref{eq:bound_Lbeta}, we have
    \begin{equation}
        \label{eq:lower_bound_lagrangian}
        \begin{aligned}
            \mathcal L_{\beta}(\mathbf x_1^t, \mathbf x_2^t, \mathbf y^t, \mathbf w^t)
             & \ge \inf_{(\mathbf x_1, \mathbf x_2, \mathbf y) \in\mathcal F} \phi(\mathbf x_1, \mathbf x_2, \mathbf y) + \frac{\beta-\zeta}{2}\|\mathbf x_1^t - \mathbf x_2^t - \mathbf y^t\|^2 \\
             & \ge \inf_{(\mathbf x_1, \mathbf x_2, \mathbf y) \in\mathcal F} \phi(\mathbf x_1, \mathbf x_2, \mathbf y) > -\infty,
        \end{aligned}
    \end{equation}
    which confirms the lower boundedness of  $\mathcal L_{\beta}(\mathbf x_1^t, \mathbf x_2^t, \mathbf y^t, \mathbf w^t)$.

    Finally, we verify the boundedness of $\left\{\mathbf{x}_1^t, \mathbf x_2^t, \mathbf y^t, \mathbf w^t\right\}$. According to the sufficient descent property \hyperlink{P2}{P2}, $\mathcal L_{\beta}(\mathbf x_1^t, \mathbf x_2^t,\allowbreak \mathbf y^t, \mathbf w^t)$ is upper bounded by  $\mathcal L_{\beta}(\mathbf x_1^0, \mathbf x_2^0,\allowbreak \mathbf y^0, \mathbf w^0)$ and so are $\phi(\mathbf x_1^t, \mathbf x_2^t, \mathbf x_1^t - \mathbf x_2^t)$ and $\|\mathbf x_1^t - \mathbf x_2^t - \mathbf y^t\|$. It follows from Proposition \ref{prop:coercivity} that $\{\mathbf x_1^t\}$ and $\{\mathbf x_2^t\}$ are both bounded. This, combined with the boundedness of $\|\mathbf x_1^t - \mathbf x_2^t - \mathbf y^t\|$ implies that $\{\mathbf y^t\}$ is bounded. Lastly, since $\mathbf w^t = \nabla h(\mathbf y^t) = \zeta \mathbf y^t$, $\{\mathbf w^t\}$ is also bounded.
\end{proof}

In the following proposition, we address the subgradient bound property \hyperlink{P3}{P3} by providing the following proposition tailored to our context, which also facilitates the determination of the stopping criterion and resource estimation.

\begin{proposition}[cf. Lemma 10 in \cite{wang2019global}]
    \label{prop:subgradient_bound}
    Let $\beta > \zeta > 0$. There exists some $\mathbf d^{+} \in \partial \mathcal L_{\beta}(\mathbf x_1^{+}, \mathbf x_2^{+}, \mathbf y^{+}, \mathbf w^{+})$ such that
    \begin{equation*}
        \|\mathbf d^{+}\| \le \left(\beta + \frac{\zeta}{\beta}\right)\|\mathbf y^+ - \mathbf y^t\|,
    \end{equation*}
\end{proposition}

\begin{proof}
    Since $\mathcal S$ is a finite discrete set, the subdifferential of the indicator function $\iota_{\mathcal S}$ is trivially $\mathbb R^n$, which means $\partial_{\mathbf x_1} \mathcal L_\beta = \mathbb R^n$. Hence, there exists $\mathbf d^{+} \in \partial \mathcal L_{\beta}(\mathbf x_1^{+}, \mathbf x_2^{+}, \mathbf y^{+}, \mathbf w^{+})$ such that
    \begin{equation}
        \label{eq:subgradient}
        \|\mathbf d^{+}\|
        \le \left\|\frac{\partial \mathcal L_\beta}{\partial \mathbf x_2}(\mathbf x_1^{+}, \mathbf x_2^{+}, \mathbf y^{+}, \mathbf w^{+})\right\|+ \left\|\nabla_{\mathbf y}\mathcal L_\beta (\mathbf x_1^{+}, \mathbf x_2^{+}, \mathbf y^{+}, \mathbf w^{+})\right\| + \left\|\nabla_{\mathbf w}\mathcal L_\beta (\mathbf x_1^{+}, \mathbf x_2^{+}, \mathbf y^{+}, \mathbf w^{+})\right\|
    \end{equation}
    Now we analyze the three terms on the right-hand side. The optimality condition for subproblem of $\mathbf x_2$ implies
    \begin{equation*}
        2\mathbf Q \mathbf x_2^+ - \lambda \boldsymbol{\mu} + \lambda\mathbf \Sigma\mathbf x_2^+ - \mathbf w^t - \beta(\mathbf x_1^+ - \mathbf x_2^+ - \mathbf y^t) = \mathbf 0.
    \end{equation*}
    Therefore, 
    % should be - \beta in the first line
    \begin{equation}
        \label{eq:subgradient_x2}
        \begin{aligned}
            \left\|\frac{\partial \mathcal L_\beta}{\partial \mathbf x_2}(\mathbf x_1^{+}, \mathbf x_2^{+}, \mathbf y^{+}, \mathbf w^{+})\right\|
             & = \|2\mathbf Q \mathbf x_2^+ - \lambda \boldsymbol{\mu} + \lambda\mathbf \Sigma\mathbf x_2^+ - \mathbf w^+ - \beta(\mathbf x_1^+ - \mathbf x_2^+ - \mathbf y^+) \| \\
             & = \|\mathbf w^t - \mathbf w^+ + \beta(\mathbf y^+ - \mathbf y^t)\|                                                                                                 \\
             & \le (\beta - \zeta)\|\mathbf y^+ - \mathbf y^t\|.
        \end{aligned}
    \end{equation}
    On the other hand, from the optimality condition for the subproblem of $\mathbf y$ and $\mathbf w$,
    \begin{equation}
        \begin{aligned}
            \label{eq:subgradient_y}
            \left\|\nabla_{\mathbf y}\mathcal L_\beta (\mathbf x_1^{+}, \mathbf x_2^{+}, \mathbf y^{+}, \mathbf w^{+})\right\| & = \|\zeta \mathbf y^+ - \mathbf w^+ - \beta(\mathbf x_1^+ - \mathbf x_2^+ - \mathbf y^+)\| \\
                                                                                                                               & = \zeta\|\mathbf y^+ - \mathbf y^t\|,
        \end{aligned}
    \end{equation}
    and
    \begin{equation}
        \begin{aligned}
            \label{eq:subgradient_w}
            \left\|\nabla_{\mathbf w}\mathcal L_\beta (\mathbf x_1^{+}, \mathbf x_2^{+}, \mathbf y^{+}, \mathbf w^{+})\right\|
             & = \|\mathbf x_1^+ - \mathbf x_2^+ - \mathbf y^+\|   \\
             & = \frac{\zeta}{\beta}\|\mathbf y^+ - \mathbf y^t\|.
        \end{aligned}
    \end{equation}
    Combining \Cref{eq:subgradient}, \Cref{eq:subgradient_x2}, \Cref{eq:subgradient_y}, and \Cref{eq:subgradient_w}, we immediately know that there exists $\mathbf d^{+} \in \partial \mathcal L_{\beta}(\mathbf x_1^{+}, \mathbf x_2^{+}, \mathbf y^{+}, \mathbf w^{+})$ such that
    \begin{equation*}
        \|\mathbf d^{+}\| \le \left(\beta + \frac{\zeta}{\beta}\right)\|\mathbf y^+ - \mathbf y^t\|.
    \end{equation*}
\end{proof}

By virtue of Proposition \ref{prop:subgradient_bound}, property \hyperlink{P3}{P3} simply holds with $C_2(\beta) = \beta + \frac{\zeta}{\beta}$. Finally, regarding the limiting continuity property \hyperlink{P4}{P4}, we formally establish it in the following proposition:
\begin{proposition}[Limiting continuity, cf. Proof of Theorem 1 in \cite{wang2019global}]
    \label{prop:limiting_continuity}
    The sequence $\left\{\mathbf{x}_1^t, \mathbf x_2^t, \mathbf y^t, \mathbf w^t\right\}$ generated by Algorithm \ref{alg:admm} satisfies the limiting continuity property \hyperlink{P4}{P4}.
\end{proposition}

\begin{proof}
    Our augmented Lagrangian $\mathcal L_\beta(\mathbf{x}_1, \mathbf x_2, \mathbf y, \mathbf w)$ defined as in \Cref{eq:lagrangian} is the sum of continuous functions and an indicator function $\iota_{\mathcal S}(\mathbf x_1)$. Since continuous functions trivially satisfy the limiting continuity property, we need only consider $\iota_{\mathcal S}(\mathbf x_1)$. However, any subsequence $(\mathbf{x}_1^{t_s}, \mathbf x_2^{t_s}, \mathbf y^{t_s}, \mathbf w^{t_s})$ of the iterates generated by Algorithm \ref{alg:admm} guarantees $\mathbf{x}_1^{t_s}\in\mathcal S$, which means $\iota_{\mathcal S}(\mathbf x_1^{t_s}) = 0$ for any $s\in\mathbb N$. On the other hand, $\mathcal S$ is a finite discrete set and hence closed. This means that $\mathbf x_1^{*} = \lim_{s \to \infty} \mathbf x_1^{t_s} \in \mathcal S$, and consequently $\iota_{\mathcal S}(\mathbf x_1^{*}) = 0$. Combining this with the limits of the continuous terms, we have
    \begin{equation*}
        \mathcal{L}_\beta(\mathbf{x}_1^*, \mathbf x_2^*, \mathbf y^*, \mathbf w^*)=\lim _{s \rightarrow \infty} \mathcal{L}_\beta(\mathbf{x}_1^{t_s}, \mathbf x_2^{t_s}, \mathbf y^{t_s}, \mathbf w^{t_s}).
    \end{equation*}
    The proposition is hereby proved.
\end{proof}

In addition to properties \hyperlink{P1}{P1}-\hyperlink{P4}{P4}, the following proposition regarding the (K\L) property of our augmented Lagrangian ensures the sequence generated by Algorithm \ref{alg:admm} converges to a unique limit point:
\begin{proposition}
    \label{prop:KL}
    The augmented Lagrangian \Cref{eq:lagrangian} is a Kurdyka-\L ojasiewicz (K\L) function (see e.g., Definition 2.2 in \cite{liCalculusExponentKurdykaLojasiewicz2018}).
\end{proposition}

\begin{proof}
    To prove $\mathcal L_{\beta}$ is a K\L \
    function, it suffices to show that it is semialgebraic \cite{bolte2007clarke, attouch2010proximal, gambella2020multiblock}. Since polynomial functions are semialgebraic, the indicator function of a semialgebraic set is semialgebraic, and the finite sum of semialgebraic functions is semialgebraic \cite{attouch2010proximal}, we only need to show that set $\mathcal S$ is semialgebraic, i.e., it can be written as the finite union of sets in the following form
    \begin{equation*}
        \{\mathbf x \in \mathbb R^n: p_i(\mathbf x) = 0,\  q_i(\mathbf x) < 0, i = 1, 2, \dots, p \},
    \end{equation*}
    where $p_i, q_i$ are real polynomial functions. This is true since
    \begin{equation*}
        \mathcal S = \bigcup_{\mathbf x_0 \in \mathcal S} \left\{\mathbf x: p_1(\mathbf x) = \sum_{i=1}^n (x_i - x_{0, i})^2 = 0, \ q_1(\mathbf x) = p_1(\mathbf x) - 1 < 0\right\}
    \end{equation*}
    and there are only a finite number of such $\mathbf x_0$'s. In conclusion, $\mathcal L_{\beta}$ is semialgebraic and consequently a K\L \ function.
\end{proof}

Synthesizing Corollary \ref{coro:sufficient_descent}, Propositions \ref{prop:boundedness} -- \ref{prop:KL}, and also Proposition 2 in \cite{wang2019global}, we have the following theorem:
\begin{thm}
    Starting from any initialization, the sequence generated by Algorithm \ref{alg:admm} converges globally to a unique limit point $\left(\mathbf{x}_1^*, \mathbf x_2^*, \mathbf y^*, \mathbf w^*\right)$, which is a stationary point of $\mathcal{L}_\beta$.
\end{thm}

Before proceeding to resource estimation, we elaborate on the stopping criterion and the selection of $\zeta$ for the ADMM algorithm. For the former, the condition on $\|\mathbf y^+ - \mathbf y^{t}\|$ monitors the dual feasibility, which entails the primal feasibility. Indeed, for the primal residual, we have
\begin{equation*}
    \begin{aligned}
        \|\mathbf x_1^+ -\mathbf x_2^+ - \mathbf y^+ \| & =  \frac{1}{\beta}\|\mathbf w^+ -  \mathbf w^{t}\|    \\
                                                        & = \frac{\zeta}{\beta}\|\mathbf y^+ - \mathbf y^{t}\|.
    \end{aligned}
\end{equation*}

Thus far, the selection of $\zeta$ has been left unspecified. However, a prudent choice of this parameter (and hence $\beta$) is essential for the practical success of the risk-parity model. If $\zeta$ is too small, the algorithm may terminate when the consistency error $\|\mathbf x_1^T - \mathbf x_2^T\|$ is too large to constitute a valid solution for the original optimization problem. In light of this, we provide an analysis of the scaling of $\zeta$ required to ensure validity at termination in following proposition:

\begin{proposition}
    \label{prop:zeta}
    There exists some constant $C_3(\lambda, \epsilon, c_1, c_2) > 0$ such that if
    \begin{equation*}
        \zeta > C_3(\lambda, \epsilon, c_1, c_2)\dfrac{n^3\sqrt{k}}{\delta},
    \end{equation*}
    then the consistency error upon the successful termination of Algorithm \ref{alg:admm} satisfies $\|\mathbf x_1^T - \mathbf x_2^T\| < \epsilon + \delta$.
\end{proposition}

\begin{proof}
    At the termination, we know that
    \begin{equation}
        \label{eq:consistency_error_1}
        \begin{aligned}
            \|\mathbf x_1^T - \mathbf x_2^T\|
             & \le \|\mathbf x_1^T - \mathbf x_2^T - \mathbf y^T\| + \|\mathbf y^T\|                             \\
             & = \frac{\zeta}{\beta}\|\mathbf y^T - \mathbf y^{T-1}\| + \frac{1}{\zeta}\|\nabla h(\mathbf y^T)\| \\
             & \le \epsilon + \frac{1}{\zeta}\|\mathbf w^T\|
        \end{aligned}
    \end{equation}

    On the other hand, according to the optimality condition for the $\mathbf x_2$ subproblem at the $T$-th iteration of Algorithm \ref{alg:admm},
    \begin{equation*}
        2\mathbf Q\mathbf x_2^T - \lambda\boldsymbol\mu + \lambda\mathbf \Sigma \mathbf x_2^T - \mathbf w^{T-1} - \beta(\mathbf x_1^T - \mathbf x_2^T - \mathbf y^{T-1}) = \mathbf 0,
    \end{equation*}\
    which implies
    \begin{equation}
        \label{eq:w^T_norm}
        \begin{aligned}
            \|\mathbf w^T\|
             & \le \|2\mathbf Q\mathbf x_2^T\| + \lambda\|\boldsymbol\mu\| + \lambda\|\mathbf \Sigma \mathbf x_2^T\| + \beta\|\mathbf y^T - \mathbf y^{T-1}\|                     \\
             & \le 2 \gamma_{\max}(\mathbf Q)\|\mathbf x_2^T\| + \lambda\sqrt{n}\|\boldsymbol \mu\|_{\infty} +  \lambda\gamma_{\max}(\mathbf \Sigma) \|\mathbf x_2^T\| + \epsilon
        \end{aligned}
    \end{equation}
    It follows from the data boundedness assumption that $\|\boldsymbol \mu\|_{\infty} < c_1$, $\gamma_{\max}(\mathbf \Sigma) < c_2 n$. For $\gamma_{\max}(\mathbf Q)$, we know that
    \begin{equation*}
        \begin{aligned}
            \gamma_{\max}(\mathbf Q)
             & = \gamma_{\max}(\boldsymbol\Sigma^\intercal \mathbf H \boldsymbol\Sigma)                                      \\
             & = \max_{\|\mathbf x\|=1} \mathbf x \boldsymbol{\Sigma}^\intercal \mathbf H\boldsymbol{\Sigma}\mathbf x        \\
             & \le \gamma_{\max}(\mathbf H)\cdot \max_{\|\mathbf x\| = 1} \mathbf x^\intercal \boldsymbol{\Sigma}^2\mathbf x \\
             & = \gamma_{\max}(\mathbf H)\gamma_{\max}^2\left(\boldsymbol{\Sigma}\right).
        \end{aligned}
    \end{equation*}
    Let $\mathbf D$ be a diagonal matrix with each diagonal term $D_{ii} = 2nx_{1, i}^+{}^2, i = 1, 2, \dots, n$, then $\mathbf H = \mathbf D - 2\mathbf x_1^+\mathbf x_1^+{}^\intercal$. Since for each $i = 1, 2, \dots, n$, $0\leq x_{1, i}^+ \leq 1 $, we have
    \begin{equation*}
        \begin{aligned}
            \gamma_{\max}(\mathbf H)
             & \le \gamma_{\max}(\mathbf D)                                          \\
             & \le \max_{\|\mathbf x\| = 1}\  \mathbf x^\intercal \mathbf D\mathbf x \\
             & = \max_{\|\mathbf x\| = 1}\   2n\sum_{i=1}^n x_{1, i}^+{}^2 x_i^2     \\
             & \le 2n \sum_{i=1}^n x_i^2 = 2n.
        \end{aligned}
    \end{equation*}
    Hence, we have that $\gamma_{\max}(\mathbf Q) \le 2n\gamma^2_{\max}\left(\boldsymbol{\Sigma}\right) = 2c_2^2n^3$. Additionally, we know that $\|\mathbf x_2^T\| \le \|\mathbf x_1^T - \mathbf x_2^T\| + \|\mathbf x_1^T\| =  \|\mathbf x_1^T - \mathbf x_2^T\| + \sqrt{k}$. Consequently, \Cref{eq:w^T_norm} leads to
    \begin{equation}
        \label{eq:w^T_norm_2}
        \|\mathbf w^T\| \le 2c_2^2n^3(\sqrt{k} +  \|\mathbf x_1^T - \mathbf x_2^T\|) + c_1\lambda\sqrt{n} + c_2\lambda n  (\sqrt{k} +  \|\mathbf x_1^T - \mathbf x_2^T\|) + \epsilon
    \end{equation}
    Combining \Cref{eq:consistency_error_1} and \Cref{eq:w^T_norm_2}, we have
    \begin{equation}
        \label{eq:consistency_error_2}
        (\zeta - 2c_2^2n^3 - c_2\lambda n ) \|\mathbf x_1^T - \mathbf x_2^T\| \le \epsilon\zeta + 2c_2^2n^3\sqrt{k} + c_1\lambda\sqrt{n} + c_2\lambda n \sqrt{k} + \epsilon.
    \end{equation}
    Let us consider the case where $\zeta > 2c_2^2n^3 + c_2\lambda n$. $\|\mathbf x_1^T - \mathbf x_2^T\| < \epsilon + \delta$ will be guaranteed if
    \begin{equation*}
        \frac{\epsilon\zeta + 2c_2^2n^3\sqrt{k} + c_1\lambda\sqrt{n} + c_2\lambda n \sqrt{k} + \epsilon}{\zeta - 2c_2^2n^3 - c_2\lambda n} < \epsilon + \delta,
    \end{equation*}
    which simplifies to the condition
    \begin{equation*}
        \zeta > \frac{2c_2^2n^3(\sqrt{k} + \epsilon + \delta) + c_1\lambda\sqrt{n} + c_2\lambda n(\sqrt{k} + \epsilon + \delta) + \epsilon}{\delta}.
    \end{equation*}
    This means that there exists some constant $C_3(\lambda, \epsilon, c_1, c_2) > 0$ such that when $\zeta > C_3(\lambda, \epsilon, c_1, c_2)\dfrac{n^3\sqrt{k}}{\delta}$, $\|\mathbf x_1^T - \mathbf x_2^T\| < \epsilon + \delta$.
\end{proof}

Based on the stopping criterion and the proper selection of $\zeta$ for our hybrid ADMM scheme, we can now evaluate the iteration complexity, which places the final piece of the puzzle in our analysis of resource usage. The result is summarized in the following theorem:
\begin{thm}\label{thm:iteration_complexity}
    For any given $0 < \epsilon, \delta < 1$, let $\zeta = \dfrac{c_4n^3\sqrt{k}}{\delta}$ for some $c_4 > C_3(\lambda,\epsilon, c_1, c_2)$ and $\beta = c_5\zeta$ for some $c_5 > \sqrt{2}$. Algorithm \ref{alg:admm} is guaranteed to converge in at most $T = \mathcal{O}\left(\dfrac{n^{6}k^{3/2}}{\epsilon^2\delta}\right)$ iterations, yielding a solution where there exists $\mathbf d^T \in \partial \mathcal L_\beta(\mathbf x_1^T, \mathbf x_2^T, \mathbf y^T, \mathbf w^T)$ such that $\|\mathbf d^T\| < \epsilon$ with the consistency error $\|\mathbf x_1^T - \mathbf x_2^T\| < \delta + \epsilon$.
\end{thm}

\begin{proof}
    We start by analyzing the lower and upper bounds of $\mathcal L_\beta$. According to Proposition \ref{prop:boundedness} and, in particular, \Cref{eq:lower_bound_lagrangian}, we have that
    \begin{equation}
        \label{eq:L_lower_bound}
        \begin{aligned}
            \min_{0 \le t \le T} \mathcal L_\beta(\mathbf x_1^t, \mathbf x_2^t, \mathbf y^t, \mathbf w^t)
             & \ge \inf_{(\mathbf x_1, \mathbf x_2, \mathbf y) \in\mathcal F} \phi(\mathbf x_1, \mathbf x_2, \mathbf y) \\
             & \ge \min_{\mathbf x_2 \in \mathbb R^n} f(\mathbf x_2)                                                                                     \\
             & = -\frac{\lambda}{2}\boldsymbol\mu^\intercal \mathbf \Sigma^{-1}\boldsymbol\mu                             \\
             & \ge -\frac{\lambda \|\boldsymbol\mu\|^2}{2\gamma_{\min}(\mathbf \Sigma)} \ge -\frac{\lambda c_1^2 n }{2c_3}
        \end{aligned}
    \end{equation}
    On the other hand,
    \begin{equation}
        \label{eq:L_upper_bound}
        \begin{aligned}
            \max_{0 \le t \le T} \mathcal L_\beta(\mathbf x_1^t, \mathbf x_2^t, \mathbf y^t, \mathbf w^t)
             & = L_\beta(\mathbf x_1^0, \mathbf x_2^0, \mathbf y^0, \mathbf w^0)                                                                                                                                                           \\
             & \le \max_{\mathbf x_2 \in \mathcal S } \left[\mathbf x_2^\intercal \mathbf Q^0\mathbf x_2  +\lambda \left(-\mathbf x_2^\intercal \boldsymbol \mu + \frac{1}{2} \mathbf x_2^\intercal \mathbf \Sigma \mathbf x_2\right)\right] \\
             & \le \gamma_{\max}(\mathbf Q^0)\|\mathbf x_2\|^2 + \lambda\left(k\|\boldsymbol \mu\|_\infty + \frac{1}{2}\gamma_{\max}(\mathbf \Sigma)\|\mathbf  x_2\|^2\right)                                                                \\
             & \le 2nk\gamma_{\max}^2(\mathbf \Sigma) + \frac{\lambda k }{2} \gamma_{\max}(\mathbf \Sigma) + \lambda k \|\boldsymbol \mu\|_\infty                                                                                          \\
             & \le 2c_2^2n^3k + \frac{c_2\lambda nk }{2} + c_1\lambda k,
        \end{aligned}
    \end{equation}
    where $\mathbf Q^0$ is the matrix $\mathbf Q$ induced by the initialization $\mathbf x_1^0$.

    Suppose that Algorithm \ref{alg:admm} terminates successfully after $T < T_{\max}$ iterations, i.e., $\|\mathbf y^T - \mathbf y^{T-1}\|  < \epsilon / (\beta + 1)$. Consequent to Proposition \ref{prop:subgradient_bound}, $\|\mathbf d^T\| < \epsilon$. On the other hand, in view of Proposition \ref{prop:zeta},  $\|\mathbf x_1^T - \mathbf x_2^T\| < \delta + \epsilon$ is guaranteed by the selection of $\zeta$. Furthermore, according to Corollary \ref{coro:sufficient_descent}, \Cref{eq:L_lower_bound}, and \Cref{eq:L_upper_bound},

    \begin{equation}\label{eq:iteration_bound_2}
        \begin{aligned}
            \left(\frac{\beta}{2} - \frac{\zeta^2}{\beta}\right)\sum_{t=1}^T ( \|\mathbf x_2^{t} - \mathbf x_2^{t-1}\|^2 + \|\mathbf y^t - \mathbf y^{t-1}\|^2 )
             & \le \max_{0 \le t \le T} \mathcal L_\beta(\mathbf x_1^t, \mathbf x_2^t, \mathbf y^t, \mathbf w^t) - \min_{0 \le t \le T} \mathcal L_\beta(\mathbf x_1^t, \mathbf x_2^t, \mathbf y^t, \mathbf w^t) \\
             & \le 2c_2^2n^3k + \frac{c_2\lambda nk }{2} + c_1\lambda k + \frac{\lambda c_1^2 n }{2c_3},
        \end{aligned}
    \end{equation}
    which means there exists some constant $c_6 > 0$ such that
    \begin{equation}
        \label{eq:termination_1}
        \|\mathbf y^T - \mathbf y^{T-1}\|  \le \sqrt{\frac{2c_5c_6}{(c_5^2-2)\zeta T}}n^{3/2}\sqrt{k}.
    \end{equation}
    On the other hand, Algorithm \ref{alg:admm} will terminate when
    \begin{equation}
        \label{eq:termination_2}
        \begin{aligned}
            \|\mathbf y^T - \mathbf y^{T-1}\|
             &
            < \frac{\epsilon}{\beta + 1}  \\
             & = \frac{\epsilon}{c_5\zeta + 1}.
        \end{aligned}
    \end{equation}

    Combining \Cref{eq:termination_1} and \Cref{eq:termination_2}, termination is guaranteed if
    \begin{equation*}
        \sqrt{\frac{2c_5c_6}{(c_5^2-2)\zeta T}}n^{3/2}\sqrt{k} < \frac{\epsilon}{c_5\zeta + 1},
    \end{equation*}
    which will be satisfied when
    \begin{equation*}
        T = \left\lceil\frac{2c_5c_6n^3k(c_5\zeta+1)^2}{(c_5^2-2)\zeta\epsilon^2} \right\rceil + 1 = \mathcal{O}\left(\frac{n^{6}k^{3/2}}{\epsilon^2\delta}\right).
    \end{equation*}
    The theorem is hereby proved.
\end{proof}

\section{Resource Analysis}\label{sec:resource}
By integrating Grover Adaptive Search for the hard constraint, we introduce an \emph{ADMM-GAS-hard} framework for solving the risk parity problem. As outlined in~Algorithm \ref{alg:admm}, we alternately solve the fixed-cardinality quadratic binary optimization in~\Cref{eq:quadratic_subproblem_quantum} by our GAS for hard constrained problems, and the convex optimizations in~\Cref{eq:admm2,eq:admm3} by classical solvers to find the solutions to the risk parity. Consider the quartic optimization problem in~\Cref{eq:risk_parity_1} subject to a fixed-crdinality constraint $k \leq n$. We aim to generate a high-quality solution characterized by a subgradient bound $\epsilon$ and a consistency tolerance $\epsilon + \delta$ in the reformulated problem~\Cref{eq:quadratic_form}. Using a GAS for hard-constrained BQP-FC with parameter $\xi = 1.34$, the framework incurs a classical overhead of $\mathcal{O}\left(\frac{n^{6}k^{3/2}}{\epsilon^2\delta}\right)$ for solving the subproblems in~\Cref{eq:admm2,eq:admm3}. The total number of quadratic oracle queries scales as:
\begin{equation}
    \label{eq:complexity_ours}
    \mathcal{O}\left(\sqrt{\binom{n}{k}}\frac{n^{6}k^{3/2} }{ \sqrt{M}\epsilon^2 \delta }\right)
\end{equation}
where $M$ denotes the degeneracy of global minimal solutions. When $k \ll n$ is a fixed constant independent of $n$, the complexity~\Cref{eq:complexity_ours} is polynomial in $n$, and we achieve exponential improvements. More generally, consider the regime where $k$ scales linearly with $n$, i.e., $k = \alpha n$. According to Lemma 7 in Chapter 10 of \cite{macwilliams1977theory}, for any $\alpha \in (0, 1)$
\begin{equation*}
    \binom{n}{k} \le \frac{1}{\sqrt{2\pi\alpha(1-\alpha)}}2^{n H_2(\alpha)}, 
\end{equation*}
where $H_2(\alpha) = -\alpha\log_2(\alpha) - (1-\alpha)\log_2(1-\alpha)$. Since $\alpha \ne \frac{1}{2}$ implies $ H_2(\alpha) < 1$, we have  
\begin{equation*}
    \begin{aligned}
        \log\left[\binom{n}{k}\operatorname{poly}(n) \right] & \le H_2(\alpha)n + \log\operatorname{poly}(n).
    \end{aligned}
\end{equation*}
Therefore, for $n$ sufficiently large, we must have $H_2(\alpha)n + \log\operatorname{poly}(n) < n$ when $\alpha \ne \frac{1}{2}$, which implies that $\binom{n}{k}\operatorname{poly}(n) \ll 2^n$. 

In comparison, the QD-GAS approach requires
\begin{equation}
    \mathcal{O}\left(\frac{2^{n/2}}{\sqrt{M}}\right)
\end{equation}
queries for quartic oracles. Thus, provide $\alpha \ne \frac{1}{2}$, for any fixed $k\leq n$, the total oracles required by ADMM-GAS-hard grows slower than QD-GAS for risk parity model. Since ADMM-GAS-hard requires only quadratic oracles, the required total gates for oracles scales as $
    \mathcal{O}\left(\sqrt{\binom{n}{k}}\frac{n^{8}k^{3/2}m }{ \sqrt{M}\epsilon^2 \delta }\right)
$ up to $\widehat{C^2R}(\theta)$ gates. However, the QD-GAS approach oracle requires $ \mathcal{O}(\frac{2^{n/2}n^4 m}{\sqrt{M}})$ up to $\widehat{C^4R}(\theta)$ gates. Regarding implementation, Claudon \textit{et al}.~\cite{Claudon2024} claimed that a $d$-controlled-Pauli-X gate, $\widehat{C^d X}$, can be decomposed with a circuit depth $\Theta\left(\log (d)^3\right)$ and gate count $\mathcal{O}(d\log (d)^4)$ using one ancilla qubit, representing the same size for $\widehat{CZ}_n$.
Given~\Cref{eq:diffusion_Unk} requires gate count $\mathcal{O}(nk)$ and circuit depth $\mathcal{O}(k\log \frac{n}{k})$, the proposed framework achieves an exponential reduction in quantum gates and depth, leading to a drastic reduction in quantum resources requirements. To illustrate the general difference in gate usage, we provide a partial oracle analysis in \Cref{tb:oracletable}.

We note, however, that ADMM is not guaranteed to converge to a global optimum for non-convex problems. Instead, the algorithm typically converges to a stationary point, which may be a local rather than global optimum. The magnitude of the optimality gap depends heavily on the objective landscape. Rigorously quantifying this gap within a hybrid quantum-classical framework poses a non-trivial theoretical challenge, and we leave this investigation for future work.

\section{Conclusion}\label{sec:conclusion}
In summary, we present a Grover search subroutine tailored for Grover Adaptive Search on hard-constrained binary quadratic programming (BQP) problems, proving that marked solutions can be found within
$
    \frac{\pi}{4}\frac{\sqrt{\binom{n}{k}}}{\sqrt{M}}
$
Grover iterations under a fixed-cardinality constraint, where \(k\) denotes the prescribed Hamming weight and \(M\) is the number of marked solutions. This yields an exponential reduction in the number of Grover iterations compared with Grover search for soft-constrained BQP, whose iteration complexity is bounded by
$
    \frac{\pi}{4}\frac{\sqrt{2^n}}{\sqrt{M}}.
$
We further provide an explicit circuit construction of the diffusion operator for hard-constrained problems based on Dicke-state preparation, together with an analysis of the circuit depth and gate complexity.

Additionally, we analyze the construction of oracles for Grover Adaptive Search applied to quadratic binary optimization and higher-degree polynomial binary programmings. Consequently, the proposed approach enables exponentially faster identification of the global optimum of a fixed-cardinality binary quadratic programming using Grover Adaptive Search, at the query-complexity level.

Furthermore, we introduce a hybrid classical--quantum ADMM framework that alternates between a classical SDP solver and our quantum solver to address a risk-parity model, which is widely used in portfolio management. Our method requires
$
    \mathcal{O}\left(\sqrt{\binom{n}{k}}\frac{n^{6}k^{3/2}m }{ \sqrt{M}\epsilon^2 \delta }\right)
$
queries to quadratic oracles in total to obtain an $\epsilon$-approximate solution. Compared with a direct implementation of a quartic oracle---which involves multi-controlled rotations such as \(\widehat{C^4R}(\theta)\)---our method yields an overall exponential reduction in quantum gate complexity and circuit depth. The framework can be further extended to higher-degree polynomial optimization problems under suitable regularity conditions, which we leave for future work.

Overall, our results are robust and provide insights into the practical application of quantum algorithms.

\section{Acknowledgment}
The authors thank Kelvin Koor for the discussion of complexity, and Yuemeng Sun for the discussion of the risk parity model.

\bibliographystyle{unsrt}
\bibliography{main} 

@article{attouch2010proximal,
  title     = {Proximal alternating minimization and projection methods for nonconvex problems: An approach based on the Kurdyka-{\L}ojasiewicz inequality},
  author    = {Attouch, H{\'e}dy and Bolte, J{\'e}r{\^o}me and Redont, Patrick and Soubeyran, Antoine},
  journal   = {Mathematics of operations research},
  volume    = {35},
  number    = {2},
  pages     = {438--457},
  year      = {2010},
  publisher = {INFORMS}
}

@article{bertsimasBestSubsetSelection2016,
  title   = {Best Subset Selection via a Modern Optimization Lens},
  author  = {Bertsimas, Dimitris and King, Angela and Mazumder, Rahul},
  year    = 2016,
  journal = {The Annals of Statistics},
  volume  = {44},
  number  = {2},
  issn    = {0090-5364},
  doi     = {10.1214/15-AOS1388}
}

@article{bolte2007clarke,
  title     = {Clarke subgradients of stratifiable functions},
  author    = {Bolte, J{\'e}r{\^o}me and Daniilidis, Aris and Lewis, Adrian and Shiota, Masahiro},
  journal   = {SIAM Journal on Optimization},
  volume    = {18},
  number    = {2},
  pages     = {556--572},
  year      = {2007},
  publisher = {SIAM}
}

@article{boyd2011distributed,
  title     = {Distributed optimization and statistical learning via the alternating direction method of multipliers},
  author    = {Boyd, Stephen and Parikh, Neal and Chu, Eric and Peleato, Borja and Eckstein, Jonathan and others},
  journal   = {Foundations and Trends{\textregistered} in Machine learning},
  volume    = {3},
  number    = {1},
  pages     = {1--122},
  year      = {2011},
  publisher = {Now Publishers, Inc.}
}

@article{DedieuLearning2021,
  title   = {Learning Sparse Classifiers: {{Continuous}} and Mixed Integer Optimization Perspectives},
  author  = {Dedieu, Antoine and Hazimeh, Hussein and Mazumder, Rahul},
  year    = 2021,
  journal = {Journal of Machine Learning Research},
  volume  = {22},
  number  = {135},
  pages   = {1--47}
}

@article{gambella2020multiblock,
  title     = {Multiblock ADMM heuristics for mixed-binary optimization on classical and quantum computers},
  author    = {Gambella, Claudio and Simonetto, Andrea},
  journal   = {IEEE Transactions on Quantum Engineering},
  volume    = {1},
  pages     = {1--22},
  year      = {2020},
  publisher = {IEEE}
}

@article{ledoitWellconditionedEstimatorLargedimensional2004,
  title     = {A Well-Conditioned Estimator for Large-Dimensional Covariance Matrices},
  author    = {Ledoit, Olivier and Wolf, Michael},
  year      = 2004,
  journal   = {Journal of Multivariate Analysis},
  volume    = {88},
  number    = {2},
  pages     = {365--411},
  issn      = {0047259X},
  doi       = {10.1016/S0047-259X(03)00096-4},
  copyright = {https://www.elsevier.com/tdm/userlicense/1.0/}
}

@article{liCalculusExponentKurdykaLojasiewicz2018,
  title   = {Calculus of the {{Exponent}} of {{Kurdyka}}--{{\L ojasiewicz Inequality}} and {{Its Applications}} to {{Linear Convergence}} of {{First-Order Methods}}},
  author  = {Li, Guoyin and Pong, Ting Kei},
  year    = 2018,
  journal = {Foundations of Computational Mathematics},
  volume  = {18},
  number  = {5},
  pages   = {1199--1232},
  issn    = {1615-3375, 1615-3383},
  doi     = {10.1007/s10208-017-9366-8}
}

@book{nesterovLecturesConvexOptimization2018,
  title     = {Lectures on {{Convex Optimization}}},
  author    = {Nesterov, Yurii},
  year      = 2018,
  series    = {Springer {{Optimization}} and {{Its Applications}}},
  volume    = {137},
  publisher = {Springer International Publishing},
  address   = {Cham},
  doi       = {10.1007/978-3-319-91578-4},
  copyright = {http://www.springer.com/tdm},
  isbn      = {978-3-319-91577-7 978-3-319-91578-4}
}

@article{Qian2005,
  author  = {Qian, Edward},
  journal = {Panagora Asset Management},
  title   = {On the financial interpretation of risk contribution: {Risk} budgets do add up},
  year    = {2005},
  doi     = {http://dx.doi.org/10.2139/ssrn.684221}
}

@book{rockafellarVariationalAnalysis1998,
  title     = {Variational {{Analysis}}},
  author    = {Rockafellar, R. Tyrrell and Wets, Roger J. B.},
  year      = 1998,
  series    = {Grundlehren Der Mathematischen {{Wissenschaften}}},
  volume    = {317},
  publisher = {Springer Berlin Heidelberg},
  address   = {Berlin, Heidelberg},
  doi       = {10.1007/978-3-642-02431-3},
  copyright = {http://www.springer.com/tdm},
  isbn      = {978-3-540-62772-2 978-3-642-02431-3}
}

@article{wang2019global,
  title     = {Global convergence of ADMM in nonconvex nonsmooth optimization},
  author    = {Wang, Yu and Yin, Wotao and Zeng, Jinshan},
  journal   = {Journal of Scientific Computing},
  volume    = {78},
  pages     = {29--63},
  year      = {2019},
  publisher = {Springer}
}

@article{Gilliam2021groveradaptive,
  doi       = {10.22331/q-2021-04-08-428},
  url       = {https://doi.org/10.22331/q-2021-04-08-428},
  title     = {Grover {A}daptive {S}earch for {C}onstrained {P}olynomial {B}inary {O}ptimization},
  author    = {Gilliam, Austin and Woerner, Stefan and Gonciulea, Constantin},
  journal   = {{Quantum}},
  issn      = {2521-327X},
  publisher = {{Verein zur F{\"{o}}rderung des Open Access Publizierens in den Quantenwissenschaften}},
  volume    = {5},
  pages     = {428},
  month     = apr,
  year      = {2021}
}

@inproceedings{10821069,
  author    = {Ollive, Robin and Louise, Stephane},
  booktitle = {2024 IEEE International Conference on Quantum Computing and Engineering (QCE)},
  title     = {Quantum Signal Processing Based Grover's Adaptative Search Oracle for High Order Unconstrained Binary Optimization Problems},
  year      = {2024},
  volume    = {02},
  number    = {},
  pages     = {256-261},
  doi       = {10.1109/QCE60285.2024.10288}
}

@inproceedings{9951196,
  author    = {Bärtschi, Andreas and Eidenbenz, Stephan},
  booktitle = {2022 IEEE International Conference on Quantum Computing and Engineering (QCE)},
  title     = {Short-Depth Circuits for Dicke State Preparation},
  year      = {2022},
  volume    = {},
  number    = {},
  pages     = {87-96},
  doi       = {10.1109/QCE53715.2022.00027}
}

@techreport{feige1997densest,
  author      = {Feige, U. and Seltser, M.},
  title       = {On the densest k-subgraph problems},
  year        = {1997},
  institution = {Weizmann Science Press of Israel},
  address     = {ISR}
}

@article{CORNEIL198427,
  title    = {Clustering and domination in perfect graphs},
  journal  = {Discrete Applied Mathematics},
  volume   = {9},
  number   = {1},
  pages    = {27-39},
  year     = {1984},
  issn     = {0166-218X},
  doi      = {https://doi.org/10.1016/0166-218X(84)90088-X},
  url      = {https://www.sciencedirect.com/science/article/pii/0166218X8490088X},
  author   = {D.G. Corneil and Y. Perl}
}

@article{markowitzPortfolioSelection1952,
  author  = {Markowitz, Harry},
  journal = {The Journal of Finance},
  title   = {Portfolio Selection},
  year    = {1952},
  month   = mar,
  pages   = {77--91},
  volume  = {7},
  doi     = {10.2307/2975974},
  langid  = {english},
  url     = {https://www.jstor.org/stable/2975974}
}

@article{10.1145/792550.792553,
  author     = {Henzinger, Monika R. and Motwani, Rajeev and Silverstein, Craig},
  title      = {Challenges in web search engines},
  year       = {2002},
  issue_date = {Fall 2002},
  publisher  = {Association for Computing Machinery},
  address    = {New York, NY, USA},
  volume     = {36},
  number     = {2},
  issn       = {0163-5840},
  url        = {https://doi.org/10.1145/792550.792553},
  doi        = {10.1145/792550.792553},
  journal    = {SIGIR Forum},
  month      = sep,
  pages      = {11–22},
  numpages   = {12}
}

@inproceedings{10.5555/1083592.1083676,
  author    = {Gibson, David and Kumar, Ravi and Tomkins, Andrew},
  booktitle = {Proceedings of the 31st International Conference on Very Large Data Bases},
  title     = {Discovering large dense subgraphs in massive graphs},
  year      = {2005},
  pages     = {721–732},
  publisher = {VLDB Endowment},
  series    = {VLDB '05},
  isbn      = {1595931546},
  location  = {Trondheim, Norway},
  numpages  = {12}
}

@inproceedings{Grover1996,
  author    = {Grover, Lov K.},
  title     = {A Fast Quantum Mechanical Algorithm for Database Search},
  booktitle = {Proceedings of the Twenty-eighth Annual ACM Symposium on Theory of Computing},
  series    = {STOC '96},
  year      = {1996},
  isbn      = {0-89791-785-5},
  location  = {Philadelphia, Pennsylvania, USA},
  pages     = {212--219},
  numpages  = {8},
  doi       = {10.1145/237814.237866},
  acmid     = {237866},
  publisher = {ACM},
  address   = {New York, NY, USA}
}

@inproceedings{Brassard1998,
  author    = {Brassard, Gilles and H{\O}yer, Peter and Tapp, Alain},
  booktitle = {Automata, Languages and Programming},
  title     = {Quantum counting},
  year      = {1998},
  address   = {Berlin, Heidelberg},
  editor    = {Larsen, Kim G. and Skyum, Sven and Winskel, Glynn},
  pages     = {820--831},
  publisher = {Springer Berlin Heidelberg},
  isbn      = {978-3-540-68681-1}
}

@article{Bulger2003,
  author    = {Bulger, D. and Baritompa, W. P. and Wood, G. R.},
  journal   = {Journal of Optimization Theory and Applications},
  title     = {Implementing Pure Adaptive Search with Grover’s Quantum Algorithm},
  year      = {2003},
  issn      = {1573-2878},
  month     = mar,
  number    = {3},
  pages     = {517--529},
  volume    = {116},
  doi       = {10.1023/a:1023061218864},
  publisher = {Springer Science and Business Media LLC}
}

@article{Baritompa2005,
  author   = {Baritompa, W. P. and Bulger, D. W. and Wood, G. R.},
  journal  = {SIAM Journal on Optimization},
  title    = {Grover's Quantum Algorithm Applied to Global Optimization},
  year     = {2005},
  number   = {4},
  pages    = {1170-1184},
  volume   = {15},
  doi      = {10.1137/040605072},
  eprint   = {https://doi.org/10.1137/040605072},
  url      = {https://doi.org/10.1137/040605072}
}

@misc{Gilliam2021,
  author        = {Austin Gilliam and Charlene Venci and Sreraman Muralidharan and Vitaliy Dorum and Eric May and Rajesh Narasimhan and Constantin Gonciulea},
  title         = {Foundational Patterns for Efficient Quantum Computing},
  year          = {2021},
  archiveprefix = {arXiv},
  eprint        = {1907.11513},
  primaryclass  = {quant-ph},
  url           = {https://arxiv.org/abs/1907.11513}
}

@misc{Durr1999,
  author        = {Christoph Durr and Peter Hoyer},
  title         = {A Quantum Algorithm for Finding the Minimum},
  year          = {1999},
  archiveprefix = {arXiv},
  eprint        = {quant-ph/9607014},
  primaryclass  = {quant-ph},
  url           = {https://arxiv.org/abs/quant-ph/9607014}
}

@inproceedings{Baertschi2019,
  author    = {B{\"a}rtschi, Andreas and Eidenbenz, Stephan},
  booktitle = {Fundamentals of Computation Theory},
  title     = {Deterministic Preparation of Dicke States},
  year      = {2019},
  address   = {Cham},
  editor    = {G{\k{a}}sieniec, Leszek Antoni and Jansson, Jesper and Levcopoulos, Christos},
  pages     = {126--139},
  publisher = {Springer International Publishing},
  isbn      = {978-3-030-25027-0}
}

@article{Nam2020,
  author    = {Nam, Yunseong and Su, Yuan and Maslov, Dmitri},
  journal   = {npj Quantum Information},
  title     = {Approximate quantum Fourier transform with O(n log(n)) T gates},
  year      = {2020},
  issn      = {2056-6387},
  month     = mar,
  number    = {1},
  volume    = {6},
  doi       = {10.1038/s41534-020-0257-5},
  publisher = {Springer Science and Business Media LLC}
}

@article{Feige2001,
  author    = {Feige, U. and Peleg, D. and Kortsarz, G.},
  journal   = {Algorithmica},
  title     = {The Dense k -Subgraph Problem},
  year      = {2001},
  issn      = {1432-0541},
  month     = mar,
  number    = {3},
  pages     = {410--421},
  volume    = {29},
  doi       = {10.1007/s004530010050},
  publisher = {Springer Science and Business Media LLC}
}

@article{Maillard2010,
  title   = {On the Properties of Equally-Weighted Risk Contributions Portfolios},
  author  = {Maillard, St{\'e}phane and Roncalli, Thierry and Te{\i}letche, J{\'e}r{\^o}me},
  journal = {Journal of Portfolio Management},
  volume  = {36},
  number  = {4},
  pages   = {60--70},
  year    = {2010},
  doi     = {10.3905/JPM.2010.36.4.060}
}

@article{Asness2012,
  title   = {Leverage Aversion and Risk Parity},
  author  = {Asness, Clifford S. and Frazzini, Andrea and Pedersen, Lasse Heje},
  journal = {Financial Analysts Journal},
  volume  = {68},
  number  = {1},
  pages   = {47--59},
  year    = {2012},
  doi     = {10.2469/faj.v68.n1.2}
}

@book{Brandes2005,
  author    = {Brandes, Ulrik},
  publisher = {Springer Science \& Business Media},
  title     = {Network analysis: methodological foundations},
  year      = {2005},
  volume    = {3418}
}

@inproceedings{Kortsarz1993,
  author    = {Kortsarz, G. and Peleg, D.},
  booktitle = {Proceedings of 1993 IEEE 34th Annual Foundations of Computer Science},
  title     = {On choosing a dense subgraph},
  year      = {1993},
  pages     = {692-701},
  doi       = {10.1109/SFCS.1993.366818}
}

@article{Bruglieri2006,
  author   = {Maurizio Bruglieri and Matthias Ehrgott and Horst W. Hamacher and Francesco Maffioli},
  journal  = {Discrete Applied Mathematics},
  title    = {An annotated bibliography of combinatorial optimization problems with fixed cardinality constraints},
  year     = {2006},
  issn     = {0166-218X},
  note     = {2nd Cologne/Twente Workshop on Graphs and Combinatorial Optimization (CTW 2003)},
  number   = {9},
  pages    = {1344-1357},
  volume   = {154},
  doi      = {https://doi.org/10.1016/j.dam.2005.05.036},
  url      = {https://www.sciencedirect.com/science/article/pii/S0166218X0500380X}
}

@article{Sotirov2020,
  author    = {Renata Sotirov},
  journal   = {Optimization Methods and Software},
  title     = {On solving the densest k-subgraph problem on large graphs},
  year      = {2020},
  number    = {6},
  pages     = {1160--1178},
  volume    = {35},
  doi       = {10.1080/10556788.2019.1595620},
  eprint    = {https://doi.org/10.1080/10556788.2019.1595620},
  publisher = {Taylor \& Francis},
  url       = {https://doi.org/10.1080/10556788.2019.1595620}
}

@article{Fratkin2006,
  author    = {Fratkin, Eugene and Naughton, Brian T. and Brutlag, Douglas L. and Batzoglou, Serafim},
  journal   = {Bioinformatics},
  title     = {MotifCut: regulatory motifs finding with maximum density subgraphs},
  year      = {2006},
  issn      = {1367-4803},
  month     = jul,
  number    = {14},
  pages     = {e150--e157},
  volume    = {22},
  doi       = {10.1093/bioinformatics/btl243},
  publisher = {Oxford University Press (OUP)}
}

@article{Ricca2024,
  author   = {Ricca, Federica and Scozzari, Andrea},
  journal  = {European Journal of Operational Research},
  title    = {Portfolio optimization through a network approach: Network assortative mixing and portfolio diversification},
  year     = {2024},
  month    = {None},
  number   = {2},
  pages    = {700-717},
  volume   = {312},
  doi      = {10.1016/j.ejor.2023.07.010},
  url      = {https://ideas.repec.org/a/eee/ejores/v312y2024i2p700-717.html}
}

@inproceedings{Saha2010,
  author    = {Saha, Barna and Hoch, Allison and Khuller, Samir and Raschid, Louiqa and Zhang, Xiao-Ning},
  booktitle = {Research in Computational Molecular Biology},
  title     = {Dense Subgraphs with Restrictions and Applications to Gene Annotation Graphs},
  year      = {2010},
  address   = {Berlin, Heidelberg},
  editor    = {Berger, Bonnie},
  pages     = {456--472},
  publisher = {Springer Berlin Heidelberg},
  isbn      = {978-3-642-12683-3}
}

@misc{yuan2024quantifyingadvantagesapplyingquantum,
  title         = {Quantifying the advantages of applying quantum approximate algorithms to portfolio optimisation},
  author        = {Haomu Yuan and Christopher K. Long and Hugo V. Lepage and Crispin H. W. Barnes},
  year          = {2024},
  eprint        = {2410.16265},
  archiveprefix = {arXiv},
  primaryclass  = {quant-ph},
  url           = {https://arxiv.org/abs/2410.16265}
}

@misc{Cuccaro2004,
  author        = {Steven A. Cuccaro and Thomas G. Draper and Samuel A. Kutin and David Petrie Moulton},
  title         = {A new quantum ripple-carry addition circuit},
  year          = {2004},
  archiveprefix = {arXiv},
  eprint        = {quant-ph/0410184},
  primaryclass  = {quant-ph},
  url           = {https://arxiv.org/abs/quant-ph/0410184}
}

@article{Gidney2018,
  author    = {Gidney, Craig},
  journal   = {{Quantum}},
  title     = {Halving the cost of quantum addition},
  year      = {2018},
  issn      = {2521-327X},
  month     = jun,
  pages     = {74},
  volume    = {2},
  doi       = {10.22331/q-2018-06-18-74},
  publisher = {{Verein zur F{\"{o}}rderung des Open Access Publizierens in den Quantenwissenschaften}},
  url       = {https://doi.org/10.22331/q-2018-06-18-74}
}

@misc{Farhi2014,
  author        = {Farhi, Edward and Goldstone, Jeffrey and Gutmann, Sam},
  month         = nov,
  title         = {A Quantum Approximate Optimization Algorithm},
  year          = {2014},
  archiveprefix = {arXiv},
  eprint        = {1411.4028},
  primaryclass  = {quant-ph},
  url           = {https://arxiv.org/abs/1411.4028}
}

@article{Cerezo2021,
  title     = {Variational quantum algorithms},
  volume    = {3},
  issn      = {2522-5820},
  url       = {http://dx.doi.org/10.1038/s42254-021-00348-9},
  doi       = {10.1038/s42254-021-00348-9},
  number    = {9},
  journal   = {Nature Reviews Physics},
  publisher = {Springer Science and Business Media LLC},
  author    = {Cerezo, M. and Arrasmith, Andrew and Babbush, Ryan and Benjamin, Simon C. and Endo, Suguru and Fujii, Keisuke and McClean, Jarrod R. and Mitarai, Kosuke and Yuan, Xiao and Cincio, Lukasz and Coles, Patrick J.},
  year      = {2021},
  month     = aug,
  pages     = {625–644}
}

@article{Kadowaki1998,
  author    = {Kadowaki, Tadashi and Nishimori, Hidetoshi},
  journal   = {Phys. Rev. E},
  title     = {Quantum annealing in the transverse Ising model},
  year      = {1998},
  month     = {Nov},
  pages     = {5355--5363},
  volume    = {58},
  doi       = {10.1103/PhysRevE.58.5355},
  issue     = {5},
  numpages  = {0},
  publisher = {American Physical Society},
  url       = {https://link.aps.org/doi/10.1103/PhysRevE.58.5355}
}

@misc{Yuan2025,
  author        = {Haomu Yuan and Daniel Stilck França and Ilia Luchnikov and Egor Tiunov and Tobias Haug and Leandro Aolita},
  title         = {Exponential Speed-ups for Structured Goemans-Williamson relaxations via Quantum Gibbs States and Pauli Sparsity},
  year          = {2025},
  archiveprefix = {arXiv},
  eprint        = {2510.08292},
  primaryclass  = {quant-ph},
  url           = {https://arxiv.org/abs/2510.08292}
}

@article{Hadfield2019,
  title     = {From the Quantum Approximate Optimization Algorithm to a Quantum Alternating Operator Ansatz},
  volume    = {12},
  issn      = {1999-4893},
  url       = {http://dx.doi.org/10.3390/a12020034},
  doi       = {10.3390/a12020034},
  number    = {2},
  journal   = {Algorithms},
  publisher = {MDPI AG},
  author    = {Hadfield, Stuart and Wang, Zhihui and O’Gorman, Bryan and Rieffel, Eleanor G. and Venturelli, Davide and Biswas, Rupak},
  year      = {2019},
  month     = feb,
  pages     = {34}
}

@article{Claudon2024,
  author    = {Claudon, Baptiste and Zylberman, Julien and Feniou, César and Debbasch, Fabrice and Peruzzo, Alberto and Piquemal, Jean-Philip},
  journal   = {Nature Communications},
  title     = {Polylogarithmic-depth controlled-NOT gates without ancilla qubits},
  year      = {2024},
  issn      = {2041-1723},
  month     = jul,
  number    = {1},
  volume    = {15},
  doi       = {10.1038/s41467-024-50065-x},
  publisher = {Springer Science and Business Media LLC}
}

@book{macwilliams1977theory,
  title     = {The Theory of Error-correcting Codes},
  author    = {MacWilliams, F.J. and Sloane, N.J.A.},
  isbn      = {9780444850102},
  lccn      = {76041296},
  series    = {Mathematical Library},
  url       = {https://books.google.com.sg/books?id=nv6WCJgcjxcC},
  year      = {1977},
  publisher = {North-Holland Publishing Company}
}

@phdthesis{yuan2026quantum,
  author = {Haomu Yuan},
  title  = {Quantum Optimisation Algorithms and their Applications},
  school = {University of Cambridge},
  year   = {2026},
  type   = {PhD thesis},
  note   = {Submitted 22 January 2026}
}

@misc{bae2026reducingcircuitresourcesgrovers,
  title         = {Reducing Circuit Resources in Grover's Algorithm via Constraint-Aware Initialization},
  author        = {Eunok Bae and Jeonghyeon Shin and Minjin Choi},
  year          = {2026},
  eprint        = {2601.17725},
  archiveprefix = {arXiv},
  primaryclass  = {quant-ph},
  url           = {https://arxiv.org/abs/2601.17725}
}

@article{ledoitHoney2023,
  author    = {Ledoit, Olivier and Wolf, Michael},
  journal   = {The Journal of Portfolio Management},
  title     = {Honey, I Shrunk the Sample Covariance Matrix},
  year      = {2004},
  issn      = {2168-8656},
  month     = jul,
  number    = {4},
  pages     = {110--119},
  volume    = {30},
  doi       = {10.3905/jpm.2004.110},
  publisher = {With Intelligence LLC}
}

\appendix
\section{Implementation of the Split \& Cyclic Shift unitary}\label{sec:scs}

Give $n\geq 2, k\geq 1$, the $\widehat{SCS}_{n,k}$ unitary can be implemented using the approach proposed in \cite{Baertschi2019}, which is based on the implementation of the $\widehat{SCS}_2^{(n)}$ and $\widehat{SCS}_{3}^{(n,t)}$ operators on in a sequential order as:
\begin{equation}
    \widehat{SCS}_{n,k} =
    \begin{cases}
        \widehat{SCS}^{(n)}_2,                                              & k = 1,     \\[4pt]
        \widehat{SCS}^{(n)}_2 \displaystyle\prod_{t=2}^{k} \widehat{SCS}_{3}^{(n,t)}, & k \geq 2 .
    \end{cases}
\end{equation}
where $\widehat{SCS}_2^{(n)}$ and $\widehat{SCS}_{3}^{(n,t)}$ is provided in Fig.~\ref{fig:SCSunitary}.

\begin{figure}[H]
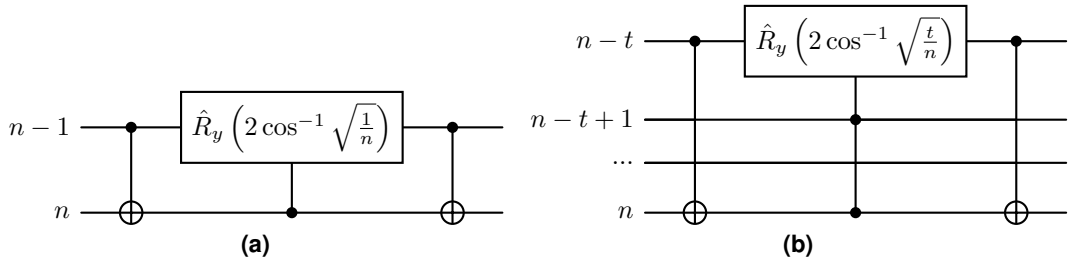

    \centering
    \begin{tabular}{c c}
        \usebox\boxSCSA                          & \usebox\boxSCSB                          \\
        \textbf{\fontfamily{phv}\selectfont (a)} & \textbf{\fontfamily{phv}\selectfont (b)}
    \end{tabular}
    \caption[The parameterised circuit of $\widehat{SCS}_2^{(n)}$ and $\widehat{SCS}_{3}^{(n,t)}$ for $\widehat{SCS}_{n,l}$ unitary.]{The circuit of $\widehat{SCS}_2^{(n)}$ and $\widehat{SCS}_{3}^{(n,t)}$ for $\widehat{SCS}_{n,l}$ unitary: (a) $\widehat{SCS}_2^{(n)}$, (b) $\widehat{SCS}_{3}^{(n,t)}$. }
    \label{fig:SCSunitary}
\end{figure}

\section{Quantum Dictionary oracle}\label{sec:quantum_dictionary_oracle}
In this section, we briefly review quantum dictionary oracle implementations for Grover Adaptive Search (GAS) applied to BPP-FC~\cite{Gilliam2021groveradaptive}.
The summary of required gates and qubits is shown in~\Cref{tb:oracletable}.

The first step is to encode the function value $f(\mathbf x)-y$ into ancilla qubits, i.e.,
\begin{equation}\label{eq:arithmetic_oracle}
    O: \ket{\mathbf x}_n\ket{0}_m \rightarrow \ket{\mathbf x}_n \ket{f(\mathbf x)-y}_m.
\end{equation}
where $\ket{\mathbf x}:=|x_1 x_2,\dots, x_i,\dots, x_n\rangle$, $x_i\in \{0, 1\}$ is one of the computational basis, and $f(\mathbf x)-y$ denotes two's-complement representation i.e. $f(\mathbf x)-y\in\mathbb{Z}, -2^{m-1}\leq f(\mathbf x) < 2^{m-1}$. Note that the coefficient of the quadratic function may not be an integer, thus we need to approximate the objective function to $m$-bit binary fractions, which will introduce an approximation error that scales as $\epsilon = \mathcal{O}(2^{-m})$.
When applying the quantum arithmetic oracle to the quadratic binary objective function in~\Cref{eq:ksubgrap}, we need to implement the $\mathcal{O}(n^2)$ additions and multiplications for building up $f(\mathbf x)-y$ on the ancilla. A direct implement the addition of $m$-bit integers on quantum arithmetic needs 2$m$+2 qubits, and $2m+1$ Toffoli gates~\cite{Cuccaro2004,Gidney2018}.
Note that, at the beginning of the multiplications and additions, we can iteratively optimize the qubits to $m'<m$ to reduce the implementation complexity; thus, the total gate required could be smaller than the one described in~\cite{Gilliam2021groveradaptive}. However, to improve this is not the focus of this article; we loosely bound the cost of implementing the quantum arithmetic oracles as $\mathcal{O}(n^2 m)$ Toffoli gates and $m$ ancilla qubits.

However, inspired by quantum data structure, quantum dictionary oracle allows to directly map the function value to ancilla qubits without performing the arithmetic operations. In summary, the quantum dictionary oracle is to construct the mapping by a controlled unitary operator as
$\widehat{CU}_\mathbf x\left(2\pi/{2^m}\right)$, where it controls on the basis $\ket{x}_n$ , then apply $\hat U_{\mathbf x}\left(2\pi/{2^m}\right)$ on the target the ancilla qubit, given by
\begin{equation}
    \hat U_{\mathbf x}\left(2\pi/{2^m}\right)\ket{K}_m = e^{i 2\pi K [f(\mathbf  x)-y]/{2^m}}\ket{K}_m.
\end{equation}
where $K$ is the index of the basis $\ket{K}_m$. Then, the action of the controlled unitary operator, including Hadamard transformation, on the initial state is given by
\begin{equation}\label{eq:qd_cux}
    \begin{aligned}
          & \widehat{CU}_\mathbf x\left(2\pi/{2^m}\right)\ket{\mathbf x}_n \hat H^{\otimes m}\ket{0}_m                          \\
        = & \ket{\mathbf x}_n \hat U_{\mathbf x}\left(2\pi /{2^m}\right)\hat H^{\otimes m}\ket{0}_m                          \\
        = & \ket{\mathbf x}_n \frac{1}{\sqrt{N'}}\sum_{K =1}^{N'-1}\hat U_{\mathbf x}\left(2\pi /{2^m}\right)\ket{K}_m  \\
        = & |\mathbf x \rangle_n  \frac{1}{\sqrt{N'}}\sum_{K =1}^{N'-1}e^{i 2\pi K [f(\mathbf x)-y]/{2^m}}\ket{K}_m
    \end{aligned}
\end{equation}
where $N' =2^m$. Then, an inverse quantum Fourier transformation, $QFT^\dagger$, is applied to the ancilla qubits to encode the binary representation of $f(\mathbf x)-y$ to the $m$ qubits as in~\Cref{eq:arithmetic_oracle}. To construct $\hat U_{\mathbf x}\left(2\pi/{2^m}\right)$ requires to apply $m$ phase gates sequentially with incremental angles, $\hat R(2^j \theta)=\left(\begin{array}{cc}
1 & 0 \\
0 & e^{i 2^j \theta}
\end{array}\right), j = 0,\dots m-1$, on each $m$ ancilla qubit. The controlled phase rotation controls on qubits corresponding to nontrivial elements in the objective function and apply $\hat R(2^j \theta)$ on the target ancilla qubits.

\AddToHook{enddocument/afteraux}{\immediate\write18{cp output.aux main.aux}}

\end{document}